\documentclass[pdflatex,sn-mathphys-num]{sn-jnl}

\usepackage{graphicx}%
\usepackage{float}
\usepackage{multirow}%
\usepackage{amsmath,amssymb,amsfonts}%
\usepackage{amsthm}%
\usepackage{mathrsfs}%
\usepackage[title]{appendix}%
\usepackage{xcolor}%
\usepackage{textcomp}%
\usepackage{manyfoot}%
\usepackage{booktabs}%
\usepackage{algorithm}%
\usepackage{algorithmicx}%
\usepackage{algpseudocode}%
\usepackage{listings}%

\theoremstyle{thmstyleone}%
\newtheorem{theorem}{Theorem}
\theoremstyle{thmstyletwo}%

\theoremstyle{thmstylethree}%
\newtheorem{definition}{Definition}%

\begin{document}

\title[Article Title]{A Semantic Link Network Model for Supporting Traceability of Logistics on Blockchain}


\author[1]{\fnm{Xiaoping} \sur{Sun}}

\author[2,3]{\fnm{Sirui} \sur{Zhuge}}

\author*[4,5]{\fnm{Hai} \sur{Zhuge}}\email{zhuge@gbu.edu.cn}

\affil[1]{\orgdiv{Key Lab of Intelligent Information Processing, Institute of Computing Technology}, \orgname{Chinese Academy of Sciences},  \city{Beijing}, \country{China}}

\affil[2]{\orgname{King’s College},\city{London}, \country{UK}}
\affil[3]{\orgname{Publicis Sapient}, \city{London}, \country{UK}}

\affil[4]{\orgname{Great Bay University},  \city{Dongguan}, \country{China}}
\affil*[5]{\orgname{Great Bay Institute for Advanced Study},  \city{Dongguan}, \country{China}}


\abstract{The ability of tracing states of logistic transportations requires an efficient storage and retrieval of the state of logistic transportations and locations of logistic objects. However, the restriction of sharing states and locations of logistic objects across organizations from different countries makes it hard to deploy a centralized database for implementing the traceability in a cross-border logistic system. This paper proposes a semantic data model on Blockchain to represent a logistic process based on the Semantic Link Network model where each semantic link represents a logistic transportation of a logistic object between two parties. A state representation model is designed to represent the states of a logistic transportation with semantic links. It enables the locations of logistic objects to be derived from the link states. A mapping from the semantic links to the blockchain transactions is designed to enable schema of semantic links and states of semantic links to be published in blockchain transactions. To improve the efficiency of tracing a path of semantic links on blockchain platform, an algorithm is designed to build shortcuts along the path of semantic links to enable a query on the path of a logistic object to reach the target in logarithmic steps on the blockchain platform. A reward-penalty policy is designed to allow participants to confirm the state of links on blockchain. Analysis and simulation demonstrate the flexibility, effectiveness and the efficiency of Semantic Link Network on immutable blockchain for implementing logistic traceability.}




\maketitle

\section{Introduction}\label{Introduction}

The traceability of a logistic system is the ability to trace the states of logistic transportations and the locations of logistic objects being transferred through the logistic transportations of a logistic process. Ensuring the traceability of logistic process can help improve customers’ informedness, especially in cross-border supply chain applications. 
To support traceability, a database for recording the states of logistic transportations and the locations of logistic objects are necessary. Due to the restrictions of sharing transportation states and locations of logistic objects among organizations from different administration regions, centralized databases are not suitable as it is hard to meet the requirements of data privacy protection of different regions. 
Blockchain provides a decentralized infrastructure for publishing and storing blockchain transactions that transfer virtual bitcoins from one payer account to another payee account \cite{RN1}.  Blockchain databases have been developed to support decentralized management of immutable transactions of user-defined virtual assets on blockchain \cite{RN2}.  Blockchain technique provides a decentralized platform for developing new logistic business models that promote visibility, traceability and transparency and reduce the cost by eliminating the third-party mediation and coordination in logistic process. Blockchain technique such as Smart Contract provides the extensibility allowing users to customize data management services on blockchain \cite{RN3}. Implementing traceability needs to record and keep updating states of logistic transportations. It needs a data model over immutable blockchain transactions to support data updating and tracing on blockchain. Two issues need to be addressed: (1) it should be able to represent the states of any logistic transportation and locations of logistic objects on blockchain and the state updates can be published in immutable blockchain transactions, and (2) it should support efficient querying of states of logistic transportations and locations of logistic objects on blockchain.
Implementing logistic state updating and tracing on blockchain transactions and smart contracts is a challenge due to the immutability of blockchain \cite{RN4}. Logistic management rules implemented in smart contract may change during transportation and should be promptly reflected in the logistic tracing process. But rewriting smart contracts is costly. The hash pointers underlying blockchain transactions can be revised to allow for updates \cite{RN5}, but it will break traceability of transactions. Verification and representation of logistic rules can be implemented separately in different smart contracts, which can partially reduce re-deployments of smart contracts \cite{RN6}. But any update made to logistic management rules will still require updating smart contracts that represents the logistic rules, which can result in a loss of traceability because it has to add new smart contracts for new logistic rules.
Publishing logistic data along with its schema on blockchain is a way to improve flexibility and adaptability. The schema should be simple, powerful, and capable of being stored on blockchain for tracing purposes. The Semantic Link Network (in short SLN) is such a flexible and extensible graph-based semantic representation model that can represent logistic transportations by semantic links published on blockchain transactions \cite{RN7}. It has schema with a graph structure to define reasoning rules among links. In the model, a logistic process can be represented by a semantic link network where nodes represent the parties and semantic links between two nodes represent logistic transportations between two parties. Each logistic object in the logistic process is uniquely represented by a binary string ID. The schema of semantic link network instances specifies rules for determining legal transitions between link states. A semantic link between two nodes is attached by a logistic object ID to represent one logistic transportation of the logistic object between two parties. The logistic transportation path corresponds to a path of links that transfer the same logistic object from one node at the beginning of the path to the node at the end of the path. A state model of link can be designed to represent the states of logistic transportations in the logistic process.  Locations of logistic objects can be derived from the semantic link states. An algorithm can be designed to publish semantic links with different states in blockchain transactions by treating the logistic objects as virtual assets of blockchain transactions. The state updating can be implemented by publishing semantic links with different states in blockchain transactions. The schema can also be published in blockchain transactions to support the verification of state transitions. The traceability can be implemented by tracing semantic links published in blockchain transactions. To further improve the query efficiency of the path of links for obtaining the state of logistic transportations in a logistic transportation path, shortcuts of a path of semantic links can be added to the semantic link path when logistic objects are transferred along the path, which makes the steps of querying a path of links within logarithmic scale. A penalty-reward policy can be designed to manage the confirmation process of states of published links in a decentralized way on blockchain platform. To demonstrate the effectiveness and the efficiency of the semantic link network model on blockchain, theoretical and experimental research is necessary.

\section{Related work}\label{sec2}
\subsection{Logistic Traceability}\label{subsec1}
A logistic process consists of multiple parties and transportations of logistic objects among parties, which requires a graph-like data model \cite{RN8} to represent and record the whole logistic process for implementing traceability. The main issue is to record the states of logistic transportation and locations of logistic objects and the updating of the logistic states and locations in a database. In the traditional supply chain management, business process models such as state transition diagrams with data flow charts can be used to represent logistic process \cite{RN9}. Logic operators are used in data flow charts to represents combinations of outputs from process nodes and the state transitions and data flows are further stored in relational databases by transforming the business process model to database schema \cite{RN10}. When modelling the cross-border logistics, centralized databases are not suitable for recording logistic state due to the restriction of sharing data flow and states among different organizations from different nations \cite{RN11}. Blockchain techniques provide an alternative solution to implement logistic traceability and enhance the privacy protection in a decentralized platform where no centralized authority is required, but it also faces many challenges in deployment of adaptability \cite{RN12}.
\subsection{Blockchain Technique}\label{subsec2}
Blockchain is a distributed digital wallet system that records each transaction of transferring an encrypted digital bitcoin from one account to another account in a decentralized way without leveraging the central authority management service \cite{RN1}. On a blockchain platform, each account owns a public-private key pair that is used to encrypt, sign and validate transactions of the owned bitcoins to other account.  A bitcoin is a virtual currency that can be split into smaller pieces of virtual currency. Bitcoins are issued when a hash ID satisfying a hard constraint is mined out. Every participant can issue a bitcoin if they can mine out such an ID. The hardness of mining a hash ID ensures that the increase of virtual currency is controllable. A blockchain transaction is a transference of a bitcoin from the payer account to the payee account, which is recorded in a data structure consisting of the transferred bitcoin number, the hash ID of the previous transaction that transferred the bitcoin to current payer account (as a pointer to connect transactions to a list) and the public key of the payee account. The transaction is encrypted by the public key of the payee account so that only the payee account can decrypt it by its private key. The transaction data structure is signed by the payer to enable the payee to verify the signature to ensure the ownership of the bitcoin of the payer. By following the pointers of the hash IDs of blockchain transactions, one can verify the bitcoin and trace back to the first transaction that issues the original bitcoin.

Blockchain transactions are broadcasted to all hosts and are stored in the latest block of the longest block chain in the local computing device of each host in the blockchain network. Each host keeps mining a hash ID for the current longest chain. The first mined out hash ID (proof-of-work) for the latest block is used as the pointer to connect the latest block to the longest block chain and a new block is prepared for accepting following new transactions. The blockchain transactions that are validated in the longest chain of blocks are accepted as legal transactions, which can avoid the double-spend on the blockchain platform because only the longest blockchain is accepted by all hosts. The hardness of computing a proof-of work prevents malicious hosts from altering the blocks in the accepted longest blockchain. Thus, transactions in blockchain are immutable.
The transactions starting from the root account of a bitcoin form a tree path from the root of the tree to a leaf node which records the path of the bitcoin transferred from the root to the leaf node and can be traced back by verifying the signature of each block transaction on the path. Smart contract techniques are extended to enable the development of various decentralized applications on blockchain by publishing services in a block so that any blockchain transaction can access the customized services of blockchain transaction management \cite{RN3}. Smart contract encapsulates executable procedures in a blockchain transaction such that the procedures can be executed on the data published in the transaction. Smart contracts are published on all peers so become immutable. The encapsulated procedure is protected and can be verified for authentication of the execution results on the data, which provides an extensible way to implement various computing services on data published in transactions  \cite{RN13}.
A blockchain platform needs to reach consensus on the longest blockchain among all hosts, which could take a long time  \cite{RN14} and high computation cost. Although blockchain databases have been receiving more and more attentions, the research focus is still mainly on the consensus and scalability issues  \cite{RN15} as there is still much time cost to reach consensus on blockchain transactions, which can greatly decrease the performance of publishing transactions and retrieval in blockchain databases  \cite{RN14}. Alternative solutions of consensus schema such as Proof-of-Stake, Proof-of-Authority and PBFT (Practical Byzantine Fault Tolerance) have been proposed to allow transactions be verified and synchronized in an efficient and decentralized way, without requiring to mine qualified hash IDs by powerful computing devices \cite{RN16, RN17}. Permissioned blockchain platforms have been developed for developing blockchain applications in controlled network environments where only authorized participants can join \cite{RN18} and most efforts are focused on confidentiality, verifiability and performance across blockchains \cite{RN19}. Blockchain databases such as BigchainDB \cite{RN2} provides enabling frameworks for developing blockchain applications without caring much on the underlying implementations of the blockchain transaction publishing, synchronization and verification. One of the major differences between the blockchain database and the traditional databases is that the blockchain database is immutable and there lacks data model schemas that can support the ACID-compliant transactions \cite{RN2}. Updating already published data in blockchain transactions should be indirectly realized by replacing the updating with appending new data published in new blockchain transactions.

\subsection{Implementing Logistic Traceability on Blockchain}\label{subsec3-3}
(1) Data storage
Distributed ledgers on blockchain are used as distributed storage shared by participants on a supply chain, where participants can update information published on the shared ledgers that are synchronized and managed in smart contracts on blockchain\cite{RN21}. A multi-blockchain framework has been proposed to enable different logistic business organizations join in and share their blockchain transactions on multiple blockchains according to a hierarchical key distribution mechanism\cite{RN22}. However, the framework does not address the logistic state representation problem. Traceability tags are used to record and verify the owners of products. It does not support logistics transportation path trace as the logistic transportation state updating is not concerned in the framework. 
Supply chain and logistics management can be implemented on the blockchain database BigchainDB, where products and goods are treated as virtual assets that are published in blockchain transactions in BigchainDB\cite{RN23}. However, the state updating issue is not addressed. 
The immutability of blockchain presents challenges when implementing a complex data management model that requires the recording and updating of logistics states and locations\cite{RN4}.  A suitable middle data model is required as a bridge between the immutable blockchain transactions and application transactions to help developer compose applications based on a middle data model\cite{RN24}, rather than directly on the raw blockchain transaction data structures, which can ease the extension and compatibility of applications and reduce the development cost and prompt adoption of blockchain techniques\cite{RN11}. Such kind of data models in previous supply chain management research have been studied in centralized database\cite{RN25}. However, the immutability and decentralization of blockchain poses new challenges for developing data model that requires recording updating of logistic states and location.
(2) Implementation of logistic rules
To implement logistic traceability, it is necessary to verify the conditions that trigger legal transportation according to specific rules for controlling the logistic process under different conditions and events, and to log transportation states while indicating the next legal states and actions, which can be implemented by if-then rules in smart contracts\cite{RN26}. Business Process Management Notation (BPMN) scripts were used to program complex workflow conditions and are published in smart contracts to facilitate the checking of conditions for each action in a supply chain, which helps implement traceability on blockchain\cite{RN27}. However, when the logistic rules are frequently updated through a logistic process, it can be costly to re-write BPMN scripts and re-deploy updated BPMN scripts in new smart contracts, which will also break the traceability. Implementing the rule interpreter of BPMN and scripts of BPMN separately in different smart contracts can reduce the cost of updating smart contracts, as the rule interpreter can remain unchanged\cite{RN6}. However, it is still inevitable to update logistic rules in a logistic process as logistic management rules in a supply chain often change. This requires handling of new logistic rules in a timely manner for an executing logistic process. Implementing tracing services in multiple abstraction layers can improve the system's adaptability and extensibility\cite{RN28}, but it often requires off-chain services to implement application level services, which can break the decentralization of blockchain and traceability.
One possible solution for managing updates of both logistic rules and logistic transportation states is to publish them with blockchain transactions. This allows for easy updating by appending new transactions to the blockchain. To achieve this, a schema for representing logistic rules and a data model for representing logistic transportation states on the blockchain are required. The semantic link network is a suitable candidate graph-like semantic model consisting of semantic nodes, semantic links, schema of reasoning rules and a semantic space where semantics of nodes and links are defined to represent concepts, instances and relationships among them. It supports both centralized applications and decentralized applications (e.g., on peer-to-peer networks) with an interactive semantics\cite{RN29}. Its nodes can be in different spaces to model cyber-physical-social systems\cite{RN30}.  It has been applied to improve e-learning and text summarization\cite{RN7, RN31}. Graph tools have been used to model logistic systems\cite{RN32}. Semantic link network model has the capability of representing different link types, which can act as a middle data model to describe the logistic transportation relationships between two parties in a logistic process. The link structure can be mapped to a blockchain transaction to support publishing logistic transportation, which ease its application in representing logistic transportations on blockchain. The schema defines semantic links and reasoning rules on semantic links of different types, which can also be published as blockchain transactions. A semantic link schema can represent logistic rules on the blockchain. Any update to logistic rules can be implemented by appending new schema on the blockchain. However, it cannot directly record the state of logistic transportations by a semantic link that published in a blockchain transaction. An extension on SLN is required for representing both the logistic transportation state and locations of logistic objects, allowing state updates to be published and queried on blockchain to implement traceability. 
(3) Trace efficiency
It is still time-consuming when the tracing process requires confirmations by participants of the logistic process along the path of links to ensure the consistence between the published state of links and the state of real logistic transportations. The decentralization of the blockchain platform enables tracing to be conducted in a concurrent way like that a decentralized query is concurrently processed in a classic P2P network\cite{RN33}.
A distributed algorithm has been proposed to accelerate blockchain tracing by executing a parallel tracing process among two nodes when the topology of blockchain transactions is a directed graph\cite{RN34}.  However, the blockchain transactions are chained to form a list and only sequential accesses are supported.  
Shortcuts can be added among nodes in structured ring P2P networks to improve the query performance to logarithmic scale\cite{RN35}. 
 (4) State confirmation
The confirmation of consistence can be conducted concurrently by following shortcuts appended on the blockchain list to speed up the whole confirmation process along the semantic link path. Each confirmation can be verified only by the direct participants of the logistic transportation. It should be conducted in a controlled way such that malicious actions can be avoided. As there is no centralized server responsible for the consistency management, a decentralized policy is required to ensure the participants to publish links consistent to the real logistic transportation states. On the blockchain platform, PoW (Proof-of-Work) takes too much computing costs to reach consensus among blocks and cannot be applied for supporting publishing semantic links for logistic traceability. PoS (Proof-of-Stake) and PBFT\cite{RN17} cannot be directly applied to two participants to confirm a link state. In \cite{RN22}, an encrypt schema is used to encrypt traceability tags for anti-counterfeit. But it does not consider the confirmation problem by participants. In \cite{RN36}, a safety checking function is implemented on smart contract to let the users of contracts to denote the safety problem of the contracts, but it does not support arguments of the consensus on the blockchain transactions. In \cite{RN37}, a penalty is given to those account that performs mis-operations and the penalty is stored in the blockchain for checking. It does not consider how to reach consensus on one transaction. A reward-penalty schema is designed to handle the bifurcation problem of blockchain transaction release in PoS based blockchain platform \cite{RN16}. The main concern for PoS is to ensure the consensus on the blocks for storing blockchain transactions published by users. A reward and penalty scheme can be designed based on the consensus status of parties to encourage nodes to reach consensus as quick as possible\cite{RN38}. However, the logistic transportation can be confirmed only by the related participants of the transportation. Thus, when implementing logistic traceability on blockchain platform, it needs a consistence management schema for confirming the consistence of published links with the real transportation by the participants.

\section{Logistic Transportation Semantic Link Network model for Supporting Traceability}\label{sec3}
\subsection{ Basic Scenario of a Logistic Process}\label{subsec1-1}

A logistic process consists of a series of logistic transportations that deliver logistic objects from their source locations to target locations, where logistic objects include the parcel to be transported and the related documents such as contracts, transportation licenses, trade certifications, and custom clearance application forms that are necessary for custom clearance across border. To record each logistic transportation for supporting traceability, we use the semantic link network model\cite{RN29} to represent a logistic transportation by a semantic link with a logistic object ID to represent a logistic transportation. 
A semantic link network consists of a set of semantic links. A semantic link takes the form: $\langle v_1:a -\gamma\rightarrow v_2:b\rangle$, where $\gamma$ is a predefined link type, $v_1$ is an instance of a semantic node of type a, $v_2$ is an instance of another semantic node of type b. $\langle v_1-\gamma\rightarrow v_2\rangle$ is a simplified format of a semantic link by omitting the type of nodes. The schema defines the legal types of the link and the node of semantic links in form of $[a-\gamma\rightarrow b]$, which defines that a link with the starting node of type a, the link of type $\gamma$  and the target node of type b is a legal link. The reasoning rules on semantic links are also defined in schema. For example, the reasoning rule defines that: a link of type  adjoin to a link of type $\alpha$ implies a link of type $\beta$, where $\alpha$ and $\beta$  ,  and $\gamma$ are link types \cite{RN7} and the operator "$.$" indicates the adjacency of two links. The semantic link network model can be used as a lightweight semantic representation model to represent logistic transportations.  
To represent logistic objects, a logistic transportation semantic link (link) $v_i - \langle l, d\rangle \rightarrow v_j$ is used for representing a logistic transportation that delivers the logistic object d from $v_i$ to $v_j$, where $d$ is the logistic object ID, $l$ indicates the link ID, $v_i$ represents the starting party of the logistic transportation and $v_j$ represents the target party of the logistic transportation. Then, the semantic links are published in blockchain transactions for implementing traceability.

\begin{figure}[ht]
\centering
\includegraphics[width=0.9\textwidth]{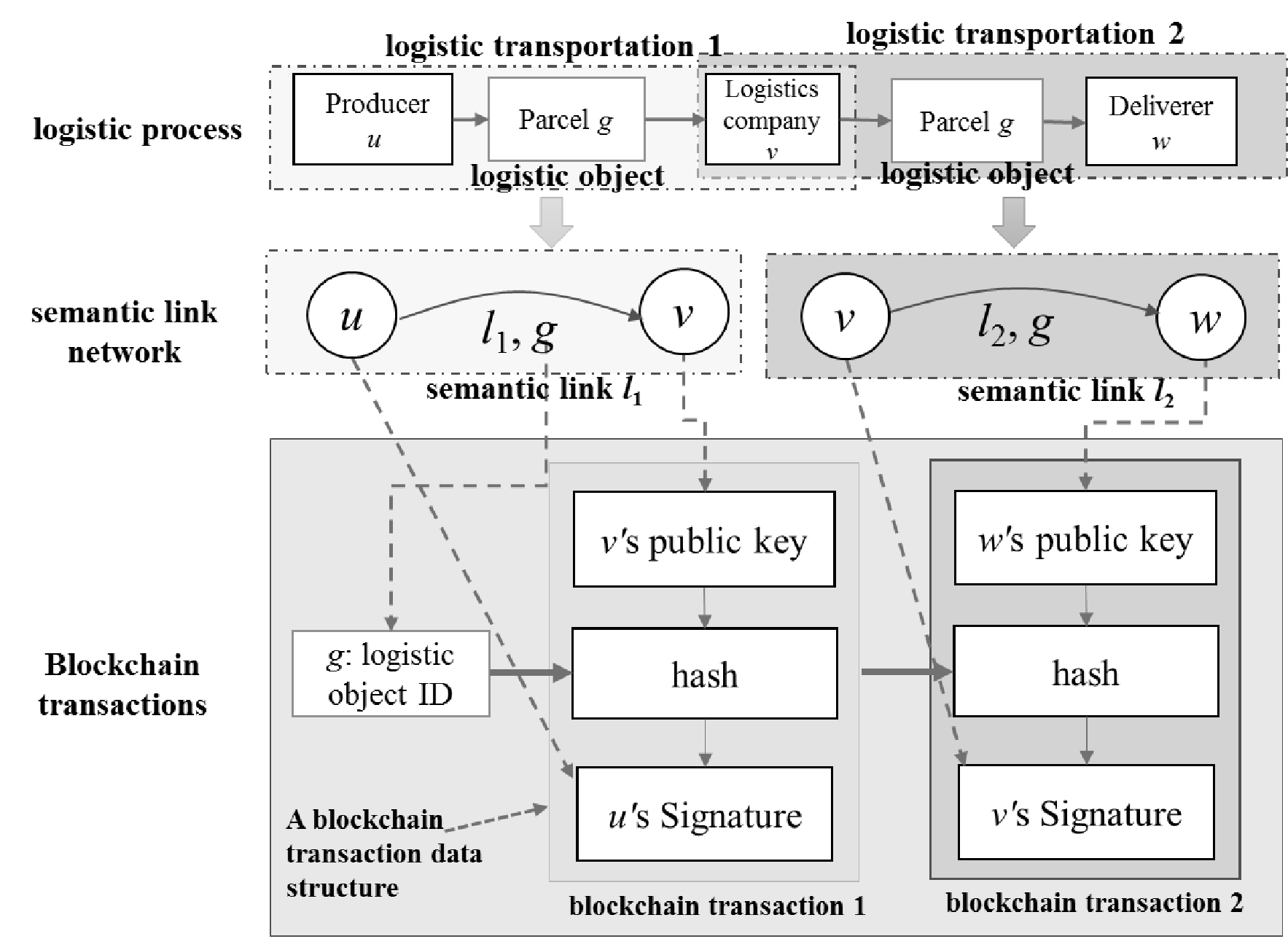}
\caption{Two logistic transportations are published as two blockchain transactions on a transaction list}\label{fig1}
\end{figure}
blockchain transaction list.
Fig. 1. shows an example consisting of two logistic transportations in a logistic process. The first transportation transfers the parcel $g$ from the producer $u$ to the logistic company $v$ and the second logistic transportation transfers $g$ from $v$ to the deliverer $w$ (a transportation company). The two logistic transportations are represented by two semantic links $l_1$ and $l_2$. Then, the two semantic links are published in two blockchain transactions to record the corresponding logistic transportations. The starting node of the link is mapped into the payer account of the blockchain transaction and the ending node of the link is mapped into the payee account of the block transaction. The logistic object ID is treated as a virtual asset ID of the blockchain transactions. One logistic transportation is recorded in one blockchain transaction. Two blockchain transactions are connected by the hash pointer of the previous one that transfers the same logistic object ID, which forms a blockchain list that can be used to trace the logistic transportations. The hash function in the first transaction block (in yellow color) is to take the payee's public key ID ($v$) and the object ID to be transferred as inputs to produce a signature that can be decoded correctly by only the receiver's private key. Then, the first transaction block can be linked to the second transaction block by inputting its hash digest and the public key of the second payee w to the hash function to produce the second signature in the second block, which ensures that only payee w can decode the second block. This chain of transactions ensures that only receiver can verify the current transaction and make the next transaction. The traceability can be conducted on each transaction on the chain by the transaction's owner.

Fig. 2 shows an example of SLN that records a logistic process $L$.  The left-hand side lists 7 logistic transportations in $L$.  The first transportation is started from the customer $v_1$ with the e-business platform $v_2$ for transferring the purchase order from $v_1$  to $v_2$.  The type of the first transportation is $place\_order$ which indicates that the platform receives an order from a customer. The logistic object $d_1$ represents the order from the customer $v_1$ to the platform $v_2$. After receiving the order, the e-business platform $v_2$ issues an order of shipment $d_2$ to the logistic company $v_3$ in the second transportation. Then, $v_3$ issues a shipment order $d_3$ to the deliverer $v_4$, which fetches the parcel $d_4$ from the warehouse $v_5$ in the fourth transportation and transport it to city A, city B and city C in following three transportations. The logistic company $v_3$ also transfers an application document $d_5$ to the custom department $v_9$ for custom clearance. The logistic transportations are mapped onto a set of links to form a SLN shown in the upper part of the right-hand side. For example, the first transportation is represented by a link $v_1 - \langle l_1, d_1 \rangle \rightarrow v_2$, where $v_1$ is the customer ID, $v_2$ represents is e-business platform ID, $l_1$ consists of the ID, the type and other attributes of the logistic transportation to be recorded such as the costs of transportation, the deadline and the constraints, and $d_1$ is the logistic object ID that represents the electronic contract document sent from $v_1$ to $v_2$.  Semantic links are published in blockchain transactions by using the logistic object ID as the virtual asset ID to be transferred in a blockchain transaction. All the logistic objects are issued by a root account that represents the whole logistic process ($T_1$ in Fig 2).
 \begin{figure}[ht]
\centering
\includegraphics[width=0.9\textwidth]{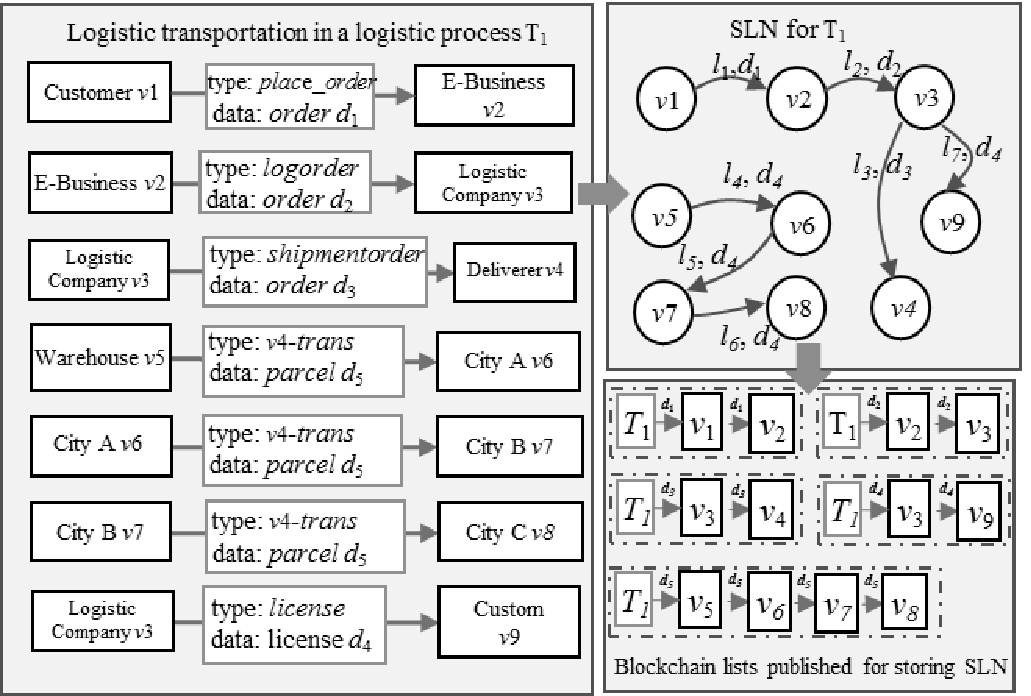}
\caption{Example of SLN on a blockchain}\label{fig2}
\end{figure}

As shown in the right-hand bottom part of Fig. 2, five blockchain lists are constructed, each having the logistic process ID $T_1$ as the root account.  The first blockchain list corresponds to the transportation of purchase order from $v_1$ to $v_2$, where the logistic object $d_1$ is the purchase order. The first arrow is the blockchain transaction transferring $d_1$  from the root account of $T_1$ to $v_1$. Then, the second blockchain transaction represents the transference of $d_1$ from $v_1$ to $v_2$, which represents the link $v_1 - \langle l_1, d_1 \rangle \rightarrow v_2$ of the first logistic transportation. The transportation of the logistic object $d_2$ is recorded in the second blockchain list where the link $v_2 - \langle l_2, d_2 \rangle \rightarrow v_3$ is published in a blockchain transaction from $v_2$ to $v_3$. Similarly, the logistic transportations of d5 forms a blockchain list with three blockchain transactions to represent three links: $v_5 - \langle l_4, d_4 \rangle \rightarrow v_6$, $v_6 - \langle l_5, d_5 \rangle \rightarrow v_7$, and $v_7 - \langle l_6, d_4 \rangle \rightarrow v_8$. Note that all blockchain lists are started from the root account $T_1$, which provides a unique access point for obtaining links in the logistic process $T_1$.
However, a logistic transportation has multiple states indicting the status of the transportation. For example, the initialized state indicates that the transportation is to be started, the transporting state indicates that the transportation is ongoing and the parcel is the path from the source to the destination, the succeeded state indicates that the transportation is successfully completed, the failed state indicates that the transportation is failed and the End state indicates that the whole transportation is terminated and the parcel will not be transported further. One transportation can be started only after its predecessor transportation is successfully completed. States of logistic transportations should be traced. Simply storing one logistic transportation in a blockchain transaction cannot reflect the state of a logistic transportation. A state model is required to represent the logistic transportation states in a logistic process\cite{RN39}. We design a state model on the SLN for representing the state of logistic transportations and mapping links to blockchain transaction according to the states of links.
 \begin{figure}[ht]
\centering
\includegraphics[width=0.9\textwidth]{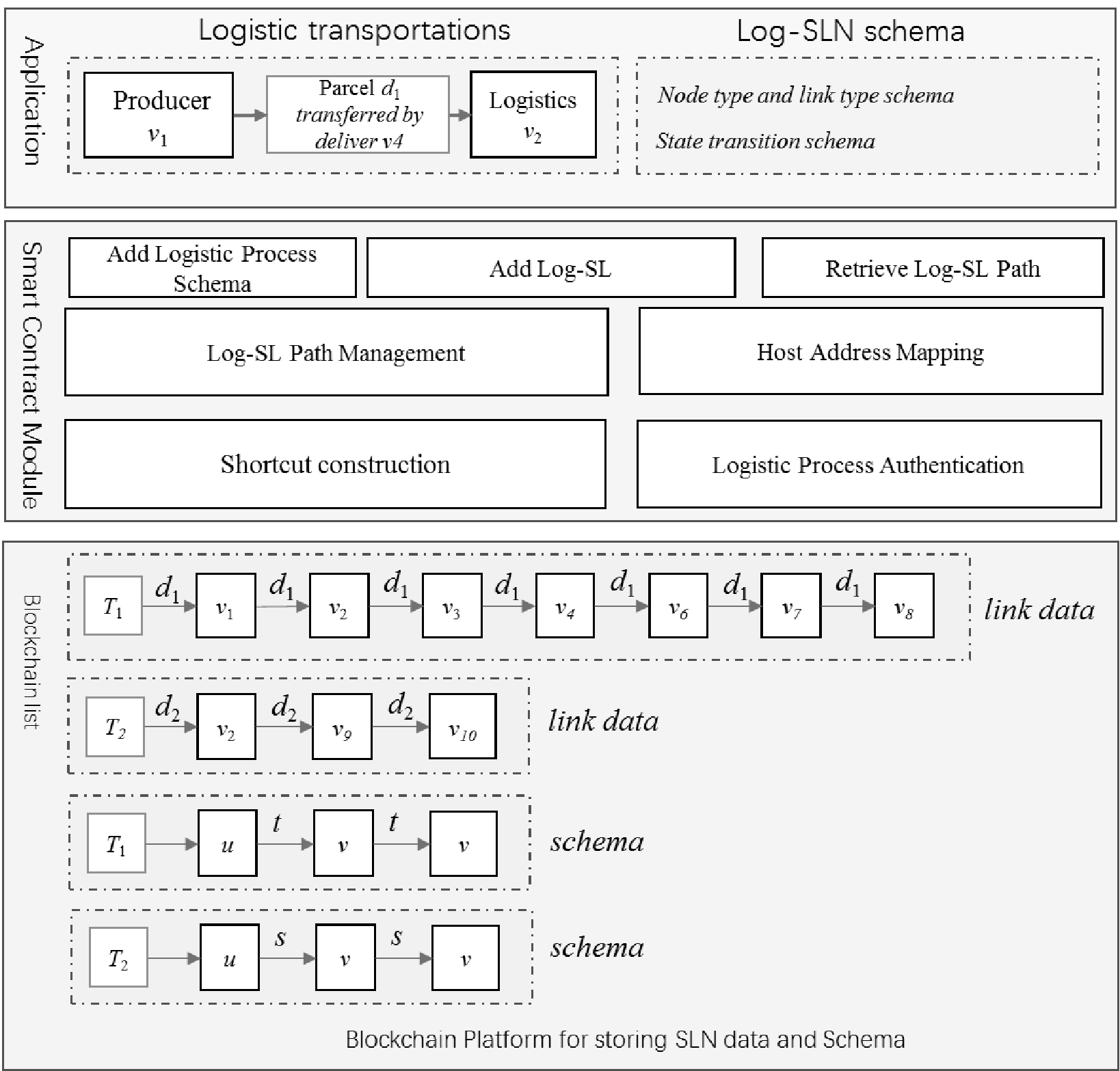}
\caption{The framework for implementing the proposed model.}\label{fig3}
\end{figure}
Fig. 3 shows the basic framework of publishing links on a blockchain platform to support tracing logistic transportations. The application layer provides the user interface for inputting and querying the transportation states and locations of logistic objects. The middle layer is based on the smart contract technique to implement link publishing, querying and validating and confirmation on the blockchain. The host address mapping module is used to determine the blockchain transaction account when publishing a link. The logistic process authentication module is used to release authentication to the legal accesses to the blockchain platform. The shortcut construction module is used to improve the efficiency of tracing the path of links by building short cut links among nodes in the SLN. The bottom layer is the blockchain platform such as BigchainDB\cite{RN2} that is used to store and access SLN data and schema by using the logistic object IDs as virtual assets of transactions.
\subsection{Semantic Link Network Model for Logistic Transportation}\label{subsec2-2}

A logistic transportation semantic link (link) uses a directed semantic link with a logistic object ID to represent a transportation of a logistic object (such as a parcel or a document) from the source party to the target party in a logistic process. A set of links forms a logistic-transportation Semantic Link Network of the logistic process. 
\begin{definition}
     A logistic transportation Semantic Link Network (SLN) is defined as: $G=\langle V,D,E\rangle$, where:
$V$ is the set of unique IDs of the nodes of the network,
$D =\{d_1, d_2, ..., d_n\}$is the set of logistic object IDs.
\(E=\{v_i-<l,d>\rightarrow v_j | v_i, v_j \in V, d \in D\}\) is the set of logistic transportation semantic links (links) where  $v_i - \langle \text{l, d} \rangle \rightarrow v_j$ is a semantic link from node $v_i$ to $v_j$ attached by a logistic object ID $d$ and $l$ is the data structure that consists of the unique ID of the link, the link type, the link state and the attached data of the logistic transportation including the cost of the transportation, the time constraints and other user-defined data structures for the logistic transportation. 
\end{definition}

A logistic process is represented by a SLN, where each party in the logistic process is uniquely represented by one node ID in $V$. The logistic object IDs are used to uniquely identify the parcels and the necessary documents to be transported. A link $v_i - \langle l,d \rangle \rightarrow v_j$ represents a logistic transportation $l$ that delivers a logistic object $d$ from starting node $v_i$ to the target node $v_j$ in $l$. 
\subsection{Location Representation of a Logistic Object}\label{subsec3}
The nodes and links in a SLN are used to represent the locations of a logistic object in a logistic process. That is, in a SLN $G=\langle V, D, E\rangle, S = V\cup E \cup \{End\}$ is the location space of logistic objects in $D$. A movement of a logistic object $d$ from one location to another location is through a link $v_i - \langle l,d \rangle \rightarrow v_j$, either from $v_i$ to the link $l$, or from the link $l$ to node $v_j$, or from $v_i$ to the End location. A logistic object $d$ can be at one of three possible locations of a link $v_i - \langle l,d \rangle \rightarrow v_j$: the starting node vi, the target node $v_j$ and the link $l$. When $d$ is at the node $v_i$, the logistic object is at the starting node and the transportation is not started yet. When $d$ is at the link $l$, it means that the logistic object is being in the transportation process. When $d$ is at the node $v_j$, the logistic object is at the target party and the transportation is completed. When all transportations of $d$ are completed, $d$ is at the End location and no more transportations of $d$ are required. 
A logistic object d is transported from its starting location to a target location through a path of n logistic transportations, which corresponds to a link path $L(d)=\{v_1 - \langle l_1,d \rangle \rightarrow v_2, v_2 - \langle l_2,d \rangle \rightarrow v_3, ..., v_{n-1} - \langle l_n,d \rangle \rightarrow v_n\}$. As a logistic object d is transported thought a path of semantic links, its locations is represented by $d(1), d(2), ..., d(n)$, where $d(n)$ is the latest location of $d$.
\subsection{Location Representation of a Logistic Object}\label{subsec4}
A link has five states: \textit{Init}, \textit{Transporting}, \textit{Succeeded}, \textit{Failed}, and \textit{End} to represent states of one logistic transportation: the initial state, the transporting state, the state of successful completion, the failure state, and the end of the logistic transportation. Definition 2 defines the way to determine the state of a link according to the transportation of the logistic object. 
\begin{definition}
    The state process of a link $l \in E$ is a set of indexed states $P_c(l)={l(j) | j \in I={1..n}}$ where $l(j): I\in \{Init, Succeeded, Failed, Transporting, End\}$ is the $j$th state of the link $l$. 
        \begin{itemize}

        \item In the state $l(j)=Init$, the logistic object $d$ of $l$ is going to move from $v_i$ to $v_k$;
        \item In the state $l(j)=Transporting$, the logistic object $d$ of $l$ is being transported through $l$;
        \item from $v_i$ to $v_k$ but the transportation is still not completed;
        \item In the state $l(j)=Succeeded$, the logistic object $d$ of $l$ is successfully moved from $v_i$ to $v_k$; 
        \item In the state $l(j)=Failed$, a logistic object $d$ of $l$ failed to move from $v_i$ to $v_k$; 
        \item In the state $l(j)= End$, a logistic object $d$ of $l$ halts.
    \end{itemize}
 \end{definition}

More states can be added according to application requirements. According to Definition 2, the location of a logistic object can be determined by inference rules.

\begin{definition}
   Location and state inference rules specify the implication relations (denoted as “$\Rightarrow$”) between the states of a link $v_k - \langle l,d \rangle \rightarrow v_j$ and the location of the logistic object $d$, where $l(n)$ denotes the current state of the link $l$ and $d(n)$ denotes the latest location of $d$ at $l$:

$l(n)=Failed \Rightarrow l(n_1) = Init | Transporting and d(n) = vj$;
    \begin{itemize}
        \item $l(n)=\textit{Init} \Rightarrow d(n) =v_k$;
        \item $l(n)=\textit{Transporting} \Rightarrow l(n_1) = \textit{Init} | \textit{Transporting}, d(n_1) = v_k | l and d(n)=l$; \textit{Failed};
        \item $l(n)=\textit{Succeeded} \Rightarrow l(n_1) = \textit{Transporting}, d(n-1) = l and d(n) = v_j$;
        \item $l(n)=\textit{End} \Rightarrow l(n_1) = \textit{Init} | \textit{Succeeded} | \textit{Failed}, d(n) = \textit{End}, and d(n+1) = \textit{End}$;
        \item $l(n)=\textit{End} \Rightarrow l(n_1) = \textit{Init} | \textit{Succeeded} | \textit{Failed}, d(n) = \textit{End}, and d(n+1) = \textit{End}$.
    \end{itemize}
 \end{definition}

\begin{figure}[ht]
\centering
\includegraphics[width=0.5\textwidth]{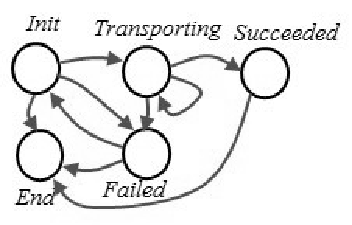}
\caption{Link state transition schema}\label{fig4}
\end{figure} 

According to the inference rules in Definition 3, the legal transitions of link states from the current state to the next state are shown in a schema. Fig. 4 illustrates the state transition graph corresponding to the schema. When the state of a transportation is in state Failed, the object ID is transferred to the payee. This is because when a transportation is confirmed to be Failed, both the payer and payee should accept that state.
\begin{definition}
    The schema of the state transition of a semantic link regulates the legal transitions as follows:
    \begin{itemize}
        \item \textit{Init} $\Leftarrow$ \textit{Transporting} $|$ \textit{Failed} $|$ \textit{End};
        \item \textit{Transporting} $\Leftarrow$ \textit{Transporting} $|$ \textit{Succeeded} $|$ \textit{Failed};
        \item \textit{Failed} $\Leftarrow$ \textit{Init} $|$ \textit{End};
        \item \textit{Succeeded}$\Leftarrow$ \textit{End};
    \end{itemize}
 \end{definition}

In the schema, $l(n) = s \Rightarrow l(n+1) = t_1 | t_2|...| t_k$ indicates that if the current state is $s$, then, the following state can be one of $t_1, t_2, ..., t_k$. The schema can be extended according to the increase of states. According to Definition 3, each state of a link has exactly one corresponding object location and the object location can be derived from the states of links. Thus, the platform designed according to the model only needs to publish and record links and their states on blockchain.
To ensure correct state transition, the semantic link network schema rules are used to define the transition schema such that the consistency of state transitions can be verified. Specifically, the semantic link network schema defines the implication rules between two links that share one node. The state transition links are used to define those implication rules among links such that only those links that can be reasoning out are legal transitions.

A schema rule of SLN is defined in the following form:
$\langle u, l, v \rangle$ and $\langle w, t, e \rangle \Rightarrow \langle u, r, e \rangle$, where $u$, $v$, $w$ and $e$ represent any node in SLN, and $l$, $t$, and $r$ represent three types of link. The schema indicates that if there are two links of type $l$and $t$ connecting nodes from $u$ to $v$ and from $w$ to $e$, then they imply a link of type $r$ from $u$ to $e$. 

A SLN schema defines a reasoning rule that can be used to derive new links from the existing links according to the link types. The state transition schema of logistic semantic links can be modeled by treating the state of a link as a type of a link. Then, a link in different states can be treated as a set of links with different types.

A semantic link schema $\langle u, s, d, v \rangle$ of SLN matches a link $v_i-<l,d>\rightarrow v_j$ with $l(n) = s$, where $s \in \{\textit{Init}, \textit{Succeeded}, \textit{Failed}, \textit{Transporting}, \textit{End}\}$.

A state transition $l(n) = s \Rightarrow l(n+1) = t$ is represented by a semantic link reasoning schema $<u, s, d, v>\Rightarrow<u, t, d, v>$. Then, a state transition schema $l(n) = s \Rightarrow l(n+1) = t_1  | t_2 | ... | t_k$ is represented by a semantic link reasoning schema: \\
$<u, s, d, v>  \Rightarrow <u, t_1, d, v> | <u, t_2, d, v> | ... | <u, t_k, d, v>$. 

The schema for describing the state transition between two different links can be defined as: \\
$<u, s, d, v> \Rightarrow <v, t, d, w>$.

For example, $<u, \textit{Succeeded}, d, v>  \Rightarrow <v, \textit{Init}, d, w>$ indicates that only when a semantic link for transferring $d$ from $u$ to $v$ is succeeded, the next link can be started.
The reasoning procedure for checking the consistence of the state transition between two links is to check if there is a rule that matches the states of the two links. For the state reachability checking, the transitive closure of the state transition graph is obtained and is used to check if two states are reachable according to the state transition schema.

Each state transition should adhere to the schema. The schema rules are published in a set of blockchain transactions from the root node T. The starting state of a rule represents the payer account and the next state represents the payee account, allowing one to obtain the schema rules by inputting the starting state. A new schema rule can be published to replace an existing schema by appending it to a blockchain transaction that follows the existing schema rule. This ensures that the latest schema rule is always held in the last blockchain transaction.

\section{Publishing and Tracing Links on Blockchain}\label{sec4}
A link is published as a blockchain transaction on a blockchain platform to implement the decentralized traceability of logistic transportations by tracing published links for transportations of a logistic object. A mapping from a link to a blockchain transaction is required to determine the payer and the payee account of the blockchain transaction for storing the link according to the state of the link as the state updating is realized by publishing new blockchain transactions. On a public blockchain platform, transactions can be verified and processed only by its owners. After a new transaction is correctly produced and published by its owner on a server, it will be broadcasted to all servers for verification and acceptance. The authentication is ensured by the transactions' signatures produced by previous transaction (see Figure 1). That is, only the owner of the signature can correctly manipulate a transaction. Usually, a client-side system is provided with user account and password management to allow users to manage its own transactions. On a blockchain platform, the first transaction is made for a newly produced bitcoin by the client software to transfer the bitcoin to it first owner (i.e., the miner). In our proposed system, the first transaction is made by the payer who own the logistic object ID and make the first logistic transportation requirement by transferring the logistic object to the next payee through a logistic process, which is recorded on a semantic link with states in a transaction. The authentication is similar to the original blockchain that uses the signature of payee to ensure the correct manipulation of a logistic transaction with a semantic link.
\subsection{Publishing Link on Blockchain}\label{subsec4-1}
Let a tuple $B=\langle a, c, b, e \rangle$ be a blockchain transaction that is published and stored on the blockchain platform to record a transference of a bitcoin value $c$ from the payer account a to the payee account $b$ with e being extra tags of the type, time, and other attribute values of the transaction such as the transaction time. According to the authentication mechanism of blockchain, when the block transaction $B$ is published on the blockchain, it is the host (the people who own the account) of the payee account $b$ which obtains the authorization (owning the encryption private key of the blockchain transaction $B$) to access and encrypt the block transaction $B$. That is, the host of the account $b$ becomes the current host of the bitcoin c and can make a next transaction to transfer $c$ to another account by applying the private key of $b$ to sign the transaction. 
A link $u-<l,d>\rightarrow v$ can be published in a blockchain transaction $B=\langle a, c, b, e \rangle$ by mapping $u$ to $a$, $d$ to $c$, $v$ to $b$ and $l(n)$ to $e$. The bitcoin $c$ of the blockchain transaction is replaced by the logistic object ID $d$, the starting node $u$ is the payer account, the target node $v$ corresponds to the payee account of the blockchain transaction and the state of $l$ is stored in $e$. To access the block transaction $B$ for obtaining the link $u-<l,d>\rightarrow v$, one needs to obtain the authorization from the host of the payee account $b$, which ensures the protection of published links. 
States of a link is published in new blockchain transactions. The location of the logistic object of a link can be used to determine the mapping from the link to a block transaction according to the states of the link. 

\begin{definition}
    A link $u-<l,d>\rightarrow v$  is published in a blockchain transaction $B=<u, d, H(l, d(n)), l(n)> $where $u$ is the payer account, the bitcoin is replaced by the logistic object ID $d$, $l(n)$ is the link state stored, and $H(l, d(n))$ is the mapping determines the payee account of the blockchain transaction that holds the link according to the location of the logistic object $d(n)$:
    \begin{itemize}
        \item $H(l, d(n)) = u\, \textit{if}\, d(n) = u$;
        \item $H(l, d(n)) = u\,  \textit{if}\, d(n) = l$;
        \item $H(l, d(n)) = v\,  \textit{if}\, d(n) = v$;
        \item $H(l, d(n)) = H(l, d(n-1))\,  \textit{if}\, d(n) = \textit{End}$, which means that the Ending state is published in the host same as the host of its previous link.
    \end{itemize}
 \end{definition}
When publishing a link in a blockchain transaction, the starting node of the link is used as the payer account of the blockchain transaction. The payee account is determined by the location of the logistic object. When the payee account is the same as payer account, the host of the payer account publishes a blockchain transaction to itself. When the logistic object is moved to the End location, it will not move further. The host of the blockchain transaction for the previous link will be responsible for publishing the last link that has the End as the target location.

Before publishing a link $u-<l,d>\rightarrow v$ for the first logistic transportation in a logistic process $T$, a root account is created by using the logistic process ID T as a root account ID on the blockchain platform. Then, a blockchain transaction $B_0=<T, d, u, \textit{Succeeded}>$ is published to transfer $d$ from the root account $T$ to the account $u$. The first blockchain transaction does not store any link. It has a default state of Succeeded. It ensures that links of logistic transportations in one logistic process are stored in blockchain lists with the same starting root account $T$. Thus, accesses to links of the same logistic process $T$ start from the same root account $T$ on the blockchain platform. Logistic transportations in different logistic processes are published in different blockchain lists. After the first blockchain transaction from the root account T is published, the blockchain transaction $B_1=<u, d, H(l, d(n)), l(n)>$ is published to store the link $u-<l,d>\rightarrow v$ according to the current link state $l(n)$ and the location $d(n)$ of $d$. The following links will be appended to the last block transaction, which forms a blockchain list for representing a transportation path of the logistic object. Each payer account of a link is the payee account of its previous link, which can ensure the correctness of the published sequence of links in the blockchain list when multiple links are published simultaneously.
The original blockchain platform does not allow double pay of one bitcoin.  However, a logistic object d can be transported to multiple target nodes from the same starting node, which requires to publish multiple links that have the same payer account, the same logistic object but the different payee accounts. For example, an application form document is sent to different departments for verification. To skip the double pay restriction of the blockchain platform, we use a suffix naming strategy to allow a logistic object d to be published in multiple blockchain transactions from one payer account to multiple payee accounts. Assuming a logistic object d at node u is transferred to k target nodes in k links:  $u-<l_1,d>\rightarrow v_1$, $u-<l_2,d>\rightarrow v_2$, ..., $u-<l_k,d>\rightarrow v_k$, then $k$ blockchain transactions are published for $B_1=<u, d.v_1, v_1, l_1(n)>$, $B_1=<u, d.v_2, v_2, l_2(n)>$, ..., $B_k<u, d.v_k, v_k, l_k(n)>$, where the $ith$ blockchain transaction $B_i = <u, d.v_i, v_i, l_i(n)>$ uses the $d.v_i$ as the logistic object ID to be transferred in the blockchain transaction from $u$ to $v_i$. When a logistic object ID with suffix is transferred in a blockchain transaction that has the same payer account and payee account, the payee account ID is appended to the suffix ID of the logistic object so that two links that are published in two blockchain transactions with the same payer and payee account will have different logistic object IDs. For example, assuming two links $u-<l_1,d>\rightarrow v_1$ and $u-<l_2,d>\rightarrow v_2$ are published with $l_1(1)= \textit{Init}$ and $l_2(1)= Init$, two blockchain transactions are published: $B_1=<u, d.v_1, u, l_1(1)>$ is published for $l_1$, and $B_2=<u, d.v_2, u, l_2(1)>$ is published for $l_2$. When the states of $l_1$ and $l_2$ are updated to Transporting, $B_3=<u, d.v_1, u, l_1(2)>$ and $B_4=<u, d.v_2, u, l_2(2)>$ are published. Appending the target node ID as the suffix to the logistic object ID of the blockchain transactions can avoid the double pay problem. Each time when a logistic object $d$ is transferred in a new link, the target node ID is appended as a suffix to the current logistic object ID in the blockchain transaction. 
\begin{figure}[htbp]
\centering
\includegraphics[width=0.8\textwidth]{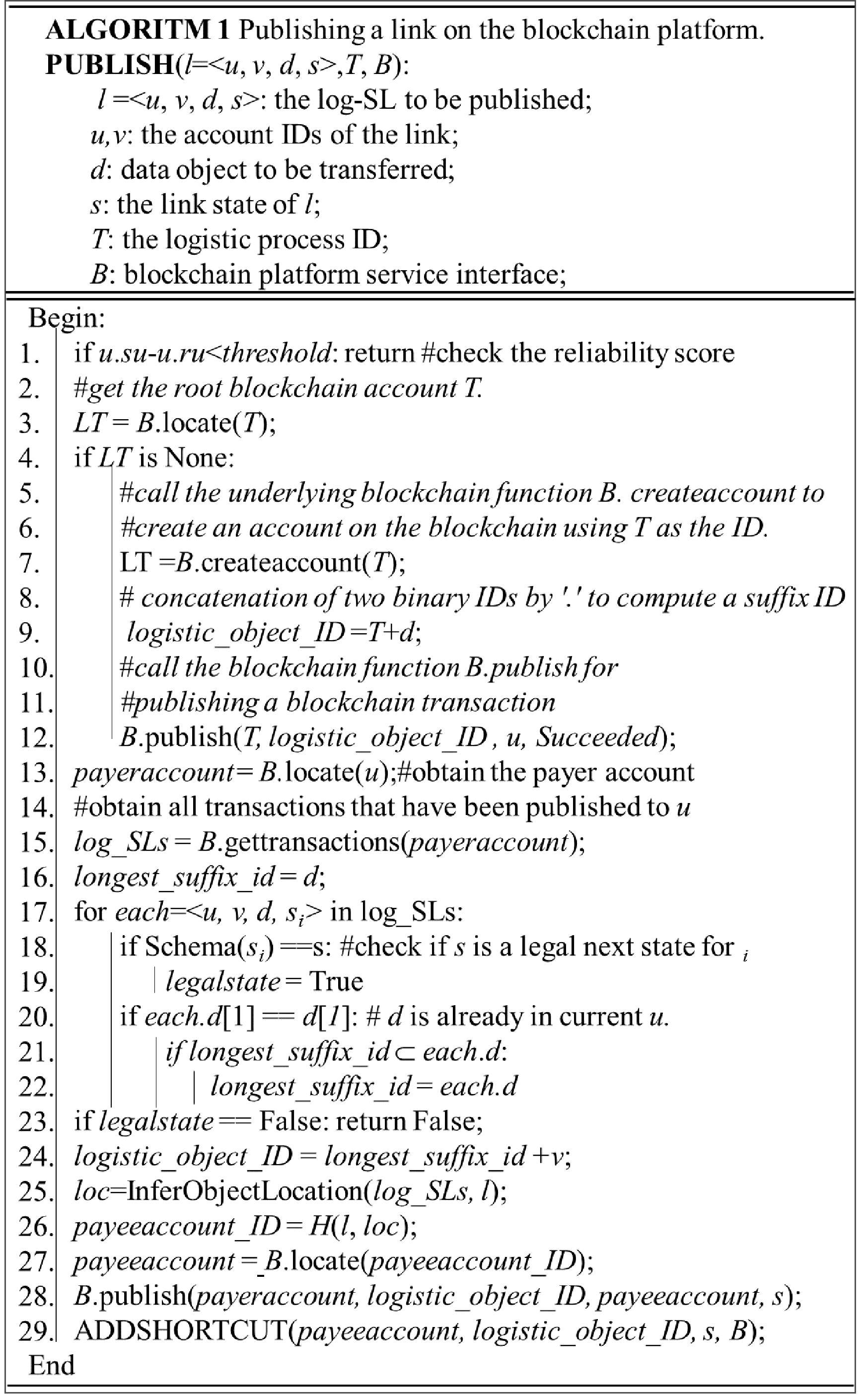}
\caption{Algorithm for publishing a link on blockchain}\label{fig5}
\end{figure}

In this way, the suffix of a logistic object $d$ is extended as $d$ is transferred through a link path, which records the sequence of the nodes that the logistic object $d$ has been transported through. For example, $d.v_1.v_2...v_k$ represents that the logistic object $d$ has been transferred through nodes $v_1$, $v_2$, ..., and $v_k$ in a sequence of links, where $v1$, $v_2$, ..., and $v_k$  are the target node of links. The suffix of a logistic object ID can be used as a record of the logistic transportation path that the logistic object has been went through. Tracing the process of a logistic object path starts from $T$ and follows the published links to verify each link until reaching the end of the path. Shortcuts from the current node to a remote node on the path can be added to improve the efficiency of the tracing process.
The algorithm for publishing a link on the blockchain platform is shown in Algorithm 1 in Fig. 5. The input arguments include the link $l$, the logistic process ID $T$ and the blockchain platform service $B$. First, the reliability score is computed to check if the account $u$ is allowed to publish a new link according to a penalty-reward policy defined in section IV.C, where $u.s_u$ and $u.r_u$ are the trustiness score and the responsibility score of $u$. Then, the root account $T$ of the logistic process is located (line 3). If there is no such a logistic process, the root account $T$ is created and the first blockchain transaction from $T$ to $u$ with logistic object $d$ is published on the blockchain (line 4 to line 12). The payer account is located for $u$ and all links that have been published from the payer account $u$ are obtained in line 15. Then, the longest suffix ID of the logistic object $d$ is obtained from line 16 to line 22 and the new suffix ID of $d$ is built by appending the ending node ID $v$ to the longest suffix ID in line 24. In this way, the logistic object $d$ will be published with the longest suffix ID that records the path of the transportation of $d$ and the longest suffix ID will be used to add shortcut links in ADDSHORTCUT in line 29. After the suffix ID is obtained, the location $loc$ of $d$ is obtained by the function InferObjectLocation through the state of the link $l$ in line 25, the payee account is located by the function $H(l, loc)$ according to the rules in Definition 5 in line 26. Finally, the blockchain transaction is published in line 28 from the payer account of u to the payee account of $v$ using the suffix ID. When checking each existing link, Schema($s_i$) is called in line 18 to verify if the current link to be published is a legal next state of an existing link. Schema($s_i$) loads the Blockchain transactions from the root node $T$ with starting account $s_i$ to determine if $s$ is its next states. If $s$ is not a legal next state, the publishing is halted in line 23.
The implementation of Algorithm 1 depends on the underlying blockchain platform to provide blockchain transaction publishing services listed in Table 1 for supporting blockchain account creation and blockchain transaction publishing and retrieval. 

\begin{table}[h]
    \caption{Functions provided by blockchain for publishing links as blockchain transactions}\label{tab1}%
    \begin{tabular}{|p{0.4\textwidth}|p{0.5\textwidth}|}
    \toprule
    Functions provided by the Blockchain platform B & Description  \\
    \midrule
    \textit{account}=B.locate(\textit{ID}) & Locate the account by its ID on the blockchain and obtain its access authority to publish a blockchain transaction using host as the payer host argument. \\
    \textit{account}=B.createhost(\textit{ID})   & Create a new account on the blockchain using ID as its host ID.   \\
    blockchainlist=B.locatechain(ID)    & Obtain blockchainlist starting from the host ID.    \\
    B.publish(\textit{payerhostid},\textit{payeehostid}, d)    & Publish a blockchain transaction from \textit{payerhostid} to \textit{payeehostid} with d being the transferred data object.    \\
    \textit{bool}=B.call(\textit{schema}, \textit{blockchainlist}, \textit{l})    &Call schema validation published as smart-contract service with address schema on the blockchain platform to validate the state of link l in a blockchain list  blockchainlist.    \\
    \textit{transactions} = B.locatetransactions(\textit{u}, \textit{v}, \textit{d})    & Obtain the blockchain transactions send from host u to host v with data object d.    \\
    \botrule
    \end{tabular}
\end{table} 
One can obtain the starting blockchain transactions of logistic objects by accessing the root account $T$ of the logistic process and then obtain links in the following blockchain transactions by tracing along the block chain lists starting from $T$. Thus, it assumes that the users who can start a logistic process should have the account $T$ and can access the account $T$ with the corresponding private key on the blockchain. The tracing process can be realized by following the hash pointers back to the payer account and use the public key of the payer account to verify the signature of the blockchain transaction. However, if one needs to check the detailed semantic link data with the real record of the logistic transportation, it needs the payee account to decrypt the blockchain transaction and obtain the semantic link data structure. Due to the lack of centralized server, one must get the authorization from the host of the payee account of each blockchain transaction to access the links and verify the semantic link and the corresponding logistic transportation, which could take times.
Fig. 6 shows an example case of publishing links on the blockchain platform for recording logistic transportation states. A logistic process $T$ is started by a logistic company $v_1$. Three logistic transportations are started from $v_1$: $l_1$ sends a parcel $d_1$ via a deliver, $l_2$ sends an application document $d_2$ from $v_2$ to the custom department $v_6$, $l_3$ sends a contract document $d_3$ to the manufactory $v_4$ and $l_4$ sends the same contract $d_3$ to the custom department $v_6$. 

\begin{figure}[htbp]
\centering
\includegraphics[width=0.9\textwidth]{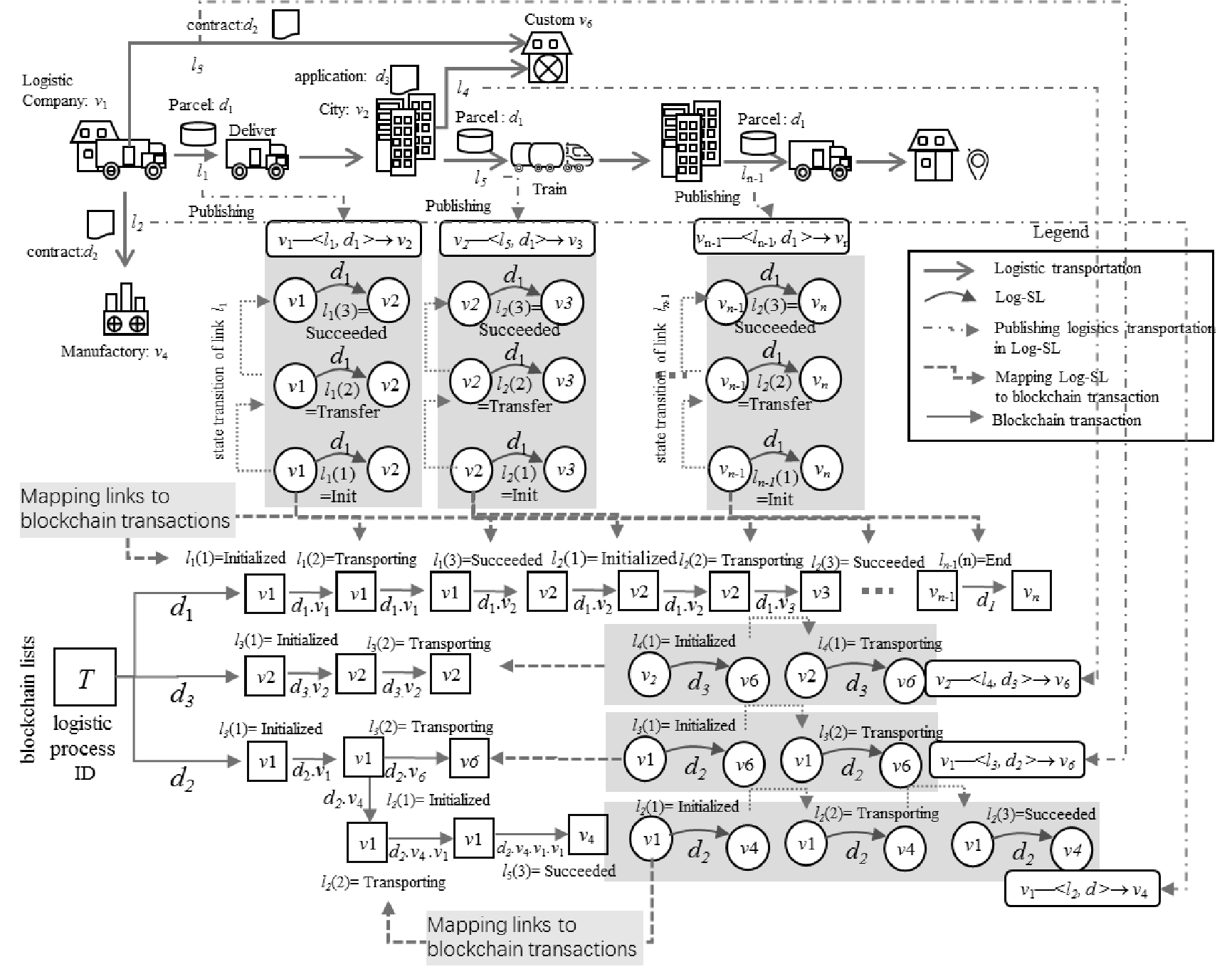}
\caption{Logistic process published on blockchain lists.}\label{fig6}
\end{figure} 

When the logistic process is initialized, a root account $T$ is first created on the blockchain platform to represent the starting account of the whole blockchain list. Logistic objects $d_1$, $d_2$ and $d_3$ are transferred from the root account $T$ in three blockchain transactions to $v_2$, $v_4$ and $v_6$, which forms three blockchain lists that allow users to trace logistic transportation paths of $d_1$, $d_2$ and $d_3$ from $T$ when the authorization of accessing $T$ is obtained.
The first logistic transportation $l_1$ for transporting a parcel $d_1$ is represented by a link $v_1-<l_1,d_1>\rightarrow v_2$. The first state of $l_1$ is Init. According to Definition 5, the payee account is still $v_1$, thus the first blockchain transaction $B_1=<v_1, d_1, v_1, l_1(1)>$ is published to store $v_1-<l_1,d_1>\rightarrow v_2$ with $l_1(1) = \textit{Init}$. After the state of $l_1$ is updated to Transporting, $v_1$ is used again as the payee account and the block transactions $B_2=<v_1, d_1, v_1, l_1(2)>$ is published to record $l_1$ with $l_1(2)$ being Transporting. When publishing $l_1$ with $l_1(3$) being Succeeded, $v_2$ becomes the payee account and the blockchain transaction from $v_1$ to $v_2$ is published as$B_3=<v_1, d_1, v_2, l_1(3)>$. After $l_1$ is completed with the Succeeded state, the next transportation from $v_2$ to $v_3$ is started. The link $v_2-<l_2, d_1>\rightarrow v_3$ is published for recording the transportation of $d_1$ from $v_2$ to $v_3$. The first state of $l_2$.is the Init state and a blockchain transaction $B_4=<v_2, d_1, v_2, l_2(1)>$ is published to store the link and its current state. Then, the next two states of $l_2$ are published in two blockchain transactions $B_5=<v_2, d_1, v_2, l_2(2)>$ and $B_6=<v_2, d_1, v_3, l_2(3)>$. Finally, the blockchain transactions for storing the links that transfer $d_1$ forms a blockchain list. 
The logistic object $d_2$ represents the application form transferred from $v_2$ to the nation custom office $v_6$ for custom clearance. The two first states of the link $l_3$ are published in two blockchain transactions from $v_1$ to $v_1$, which forms the blockchain list starting from $d_2$. 
The logistic transportation $l_3$ and $l_4$ correspond to the transportation of the contract represented by $d_2$ from $v_1$ to the manufacturer $v_4$ and from $v_2$ to the custom office $v_6$ respectively. They are recorded in two links. The branch in the blockchain list represents that the logistic object $d_2$ has multiple targets. The suffix naming schema is applied to give the logistic object $d_2$ two different IDs $d_2.v_4$ and $d_2.v_6$ when publishing the two links from $v_1$ to $v_6$ and $v_7$.

\subsection{Traceability of Semantic Links on Blockchain}\label{subsec4-3}
After publishing links in blockchain transactions, tracing logistic transportations can be implement by collecting the states of links of the logistic transportations published on blockchain platform. Different levels of traceability can be supported with different costs for collecting states of links published for recording logistic transportations.
(1) \textit{Transportation state traceability}.  Checking the current state of a logistic transportation requires to obtain a link and check its latest state by locating the blockchain transaction published for recording the latest state of the link. For example, checking if a logistic transportation of a parcel to a customer is Succeeded or Failed can be implemented by finding the link that has the target node being the customer and checking whether the link state is Succeeded. This can be implemented by locating the payee account of the blockchain transaction that holds the link and obtaining the links that has been published by the payee account.
(2) \textit{Path state traceability}.  Checking the states of a set of logistic transportations that a logistic object has transported through. A link path corresponds to such a set of logistic transportations published in blockchain transaction in one blockchain list. Implementing the path state traceability requires to obtain the sequence of links and their corresponding states, which needs to obtain all blockchain transactions in a blockchain list from the root account of the logistic process.
 (3) \textit{Logistic process traceability}.  This requires obtaining all blockchain lists that are published for recording transference paths of all objects in a logistic process. 
Different levels of traceability require to obtain different numbers of blockchain transactions. One of the major challenges is to improve the efficiency of tracing the path of a logistic object. The natural way is to traverse the whole block transactions for obtaining all links on a path, which could be time consuming because checking semantic links stored in a blockchain transaction requires obtaining the authority of the hosts of the payee account of blockchain transactions.  To make the process more efficient, one way is to reduce the number of accesses to accounts and another way is to let the tracing process be conducted in parallel.
\subsection{Reducing the Average Accesses to Account}\label{subsec4-4}
To reduce the number of accesses to the hosts while still maintaining a certain level of decentralization, it is necessary to keep balance between the number of accesses to the hosts and the decentralization of storage. According to the location mapping in Definition 5, in two cases of the location of $d$, $u-<l_1,d>\rightarrow v_1$ is mapped onto the payer account $u$ and in one case $l$ is mapped onto the payee account $v$.  Due to the location inference rules in Definition 3, the hosts of the blockchain transactions published for a link with different states is unevenly distributed among two computing host of the two accounts. Thus, for obtaining the states of links, the average number of accessing the two hosts is reduced from 2 to 1.4 as shown in Theorem 1.

\begin{theorem}[Theorem subhead]\label{thm1}
Assuming that a query on the k states of a given link $l$ visits each state with the equal probability, the expected times of accessing the hosts of blockchain transactions for $l$ is 1.414, when using the host mapping function in Definition 5 
\end{theorem}
\begin{proof}
From the link inference rules in Definition 3 and location mapping function in Definition 5, one can see that 5 out of 6 states of a link $u-<l_1,d>\rightarrow v_1$ is mapped into the account $u$. Then, if all $k$ states are within the 5 states stored in $u$, then one access to the account $u$ is enough, if there is at least one visit to the state store in $v$, then two accesses are required. The expectation can be obtained by modelling the probability of having one host access and having two hosts accesses as a binomial distribution: \\

$E\left(X\right)=\frac{1}{n}\sum_{i=1}^{n-1}{\left(p^i+\left(1-p\right)^i+\left(1-p^i-\left(1-p\right)^i\right)\times2\right),}$\\

where n=6 and p=5/6.

That is, when the states are stored in an unbalanced way distributed on two hosts of a link, more states of links can be accessed on one host of a transaction that stores more states of links than another. When $p=5/6$, $E(X)$ obtains the minimum.\ 

\end{proof}
\subsection{Improving Efficiency of Path Tracing}\label{subsec4-5}
Although the unbalanced storage distribution of link states can reduce the average accessing times for checking the states of a link, it is still inevitable to traverse each host of account of each blockchain transaction published for a link path when implementing the link path traceability. The total number of accesses is proportional to the number of hosts for storing the path of links. Specifically, to access a path $L(d_i)={v_1, l_1, v_2, l_2, v_3, ..., l_n, v_n}$ of the object di from node $v_1$ to node $v_n$, it first needs to visit the root node to obtain the first blockchain transaction starting from $v_1$ and then it obtains the first link $l_1$ and its corresponding states. Following the blockchain transaction pointer from $v_1$, the next host $v_2$ is accessed to obtain the second link $l_2$ until reaching the final node $v_n$, where all hosts $v_1$, $v_2$, $v_3$, ..., and $v_n$ are accessed sequentially by at least one time and each access needs an authorization of the payer account of the blockchain transaction. It can be time-consuming when sequentially visiting each payee account for verification of blockchain transactions.

To reduce the number of accesses to hosts for getting a link path $L(d)$, a simple way is to separate the whole path $L(d$) into to $\left\lceil n/s \right\rceil$ segments and select the first node in each segment of $s$ nodes as a representative host to hold the links within the segment so that one access to the representative node can obtain links of the whole segment of $L(d)$.  $\left\lceil n/s \right\rceil$ representative nodes are enough to hold all n links and the number of accesses of hosts is reduced to  $\left\lceil n/s \right\rceil$.  However, the total length n of the path cannot be obtained during the link publishing process, which makes it difficult to obtain a proper s for determining $\left\lceil n/s \right\rceil$ . To overcome this limitation, an algorithm is designed to add shortcuts among nodes on a path of link $L(d)$ in a probabilistic manner without knowing the length of the path.  Instead of splitting the path L(d) into segments, it lets each node on a link path randomly add a set of shortcut links from itself to its following nodes along the path when a link is being published in a blockchain transaction. A shortcut link allows a query along the path $L(d)$ to jump over a sub sequence of nodes from a node to the rear part of the link path by following the pointer of the shortcut link. The aim is to allow the tracing process to visit $log(n)$ nodes to reach the end of the path and obtain all links in $L(d)$ without knowing n. After building the shortcuts, the node holds replicas of node IDs within the sub-sequence of the link path from the node to its longest shortcut on the link path so that one can obtain the path in $L(d)$ by visiting $log(n)$ nodes. In this way, $log(n)$ shortcuts are built for each node to enable accessing the path $L(d)$ within $log(n)$ steps of jumps from the root node to the final target node in $L(d)$. Since each node is stored in one host, thus, $log(n)$ hosts are accessed to obtain the link path.

\begin{figure}[htbp]
\centering
\includegraphics[width=0.5\textwidth]{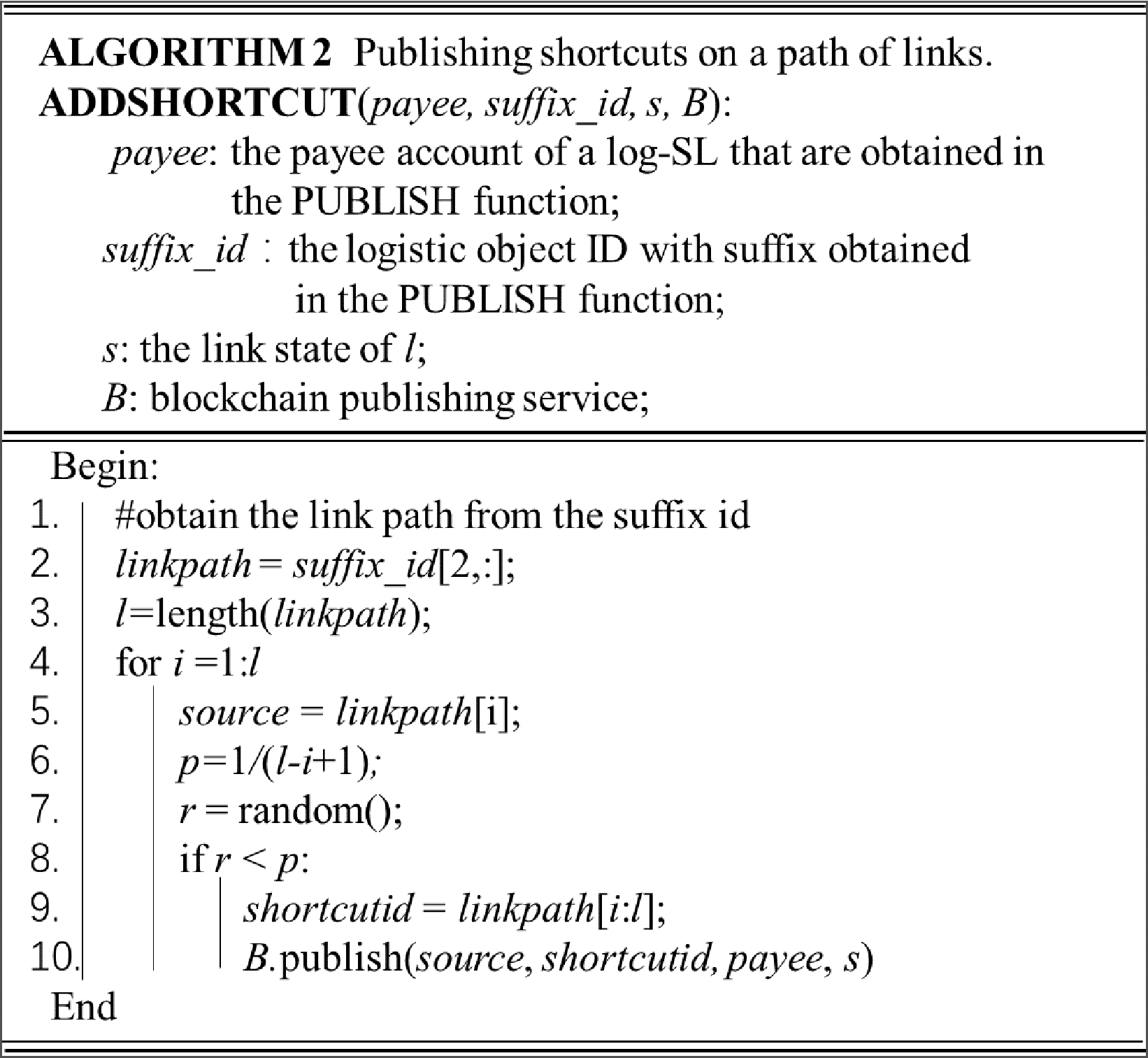}
\caption{Publishing shortcut links.}\label{fig7}
\end{figure} 

Algorithm 2 in Fig.7 shows the shortcut publishing algorithm. A shortcut is constructed as the logistic object $d$ is transferred to a payee account when a link is published in a block transaction from a payer account to the payee account. The suffix ID of the object $d$ is used as a token to record the account IDs that it has been transferred to. Thus, the path of $d$ can be obtained from the suffix ID (in Line 1). For each node $v_i$ on the path, the distance si from $v_i$ to payee is $l-i+1$ and a short cut is constructed in a random sampling way with the sampling probability of $1/s_i$ (from line 6 to line 10). When a shortcut link is built, a blockchain transaction with the suffix ID of the logistic object $d$ from $v_i$ to payee is published by $v_i$ so that the account $v_i$ will have a pointer to payee (line 10). Note that the length $l$ is not the final length of the link path, it is the length of current path that has been passed through by $d$. Thus, the algorithm does not need to know the length of the final path.
 According to Theorem 3 and Theorem 4, each node along the path $L(d_i)$ will have$ O(ln(n))$ shortcut links on average and the expected steps to obtain the whole links of $L(d_i)$ is $O(ln(n))$.
Fig. 8 shows an example of shortcuts appended on a blockchain list for one logistic object $d_1$. There are shortcut links across multiple nodes on the link path.  For example, $v_1$ has two shortcuts to $v_3$ and $v_5$ respectively, then it has replicas of node IDs of the path from $v_1$ to $v_5$. The shortcut to $v_3$ is added with a probability of $P(\textit{add})=1/2$ because logistic object $d_1$ moves 2 steps from $v_1$. The shortcut to $v_5$ is added with a probability of $P(\textit{add})=1/4$.  By accessing $v_1$, one can obtain a set of nodes from $v_2$ to $v_5$ of the current path, which can greatly speed up traversal performance.

\begin{figure}[htbp]
\centering
\includegraphics[width=0.5\textwidth]{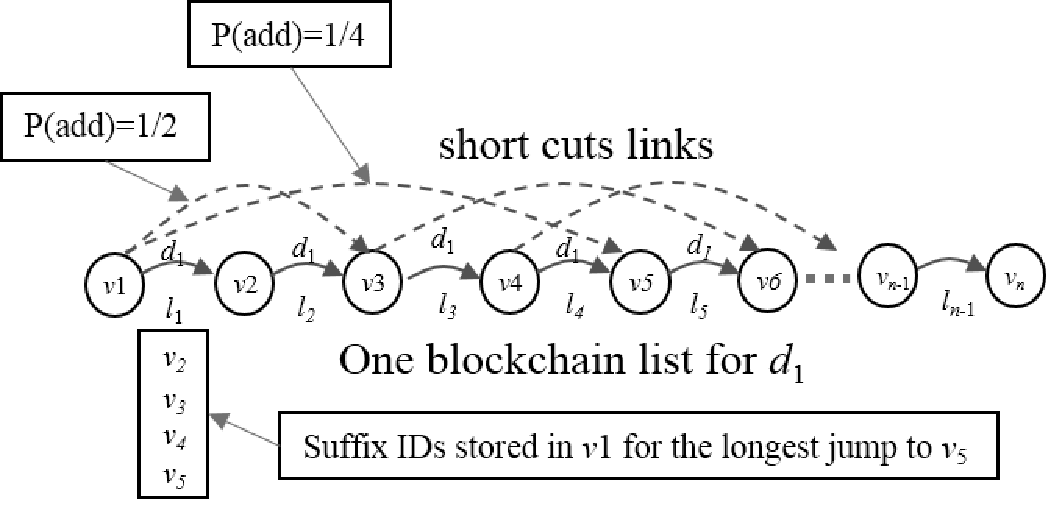}
\caption{Shortcuts on a blockchain list.}\label{fig8}
\end{figure} 

A query on a path of links is processed in an iterative way to obtain a tree structure formed by the paths of a logistic object $d$ transported from the root account $T$ (Shown by Algorithm 3 in Fig. 9). The Algorithm 3 includes following input arguments: $\textit{parent}$ records the current tree node to be constructed to record the paths starting from the first root node of the tree used as the final output argument of the algorithm, $\textit{visist}$ list records the account IDs that have been processed, \textit{$currenttrepath$}  records the tree path that has been traversed from the root account, T is the current blockchain account ID that the query is processing, the logistic object ID $d$ is the query target ID and B is the blockchain platform service interface.

\begin{figure}[htbp]
\centering
\includegraphics[width=0.5\textwidth]{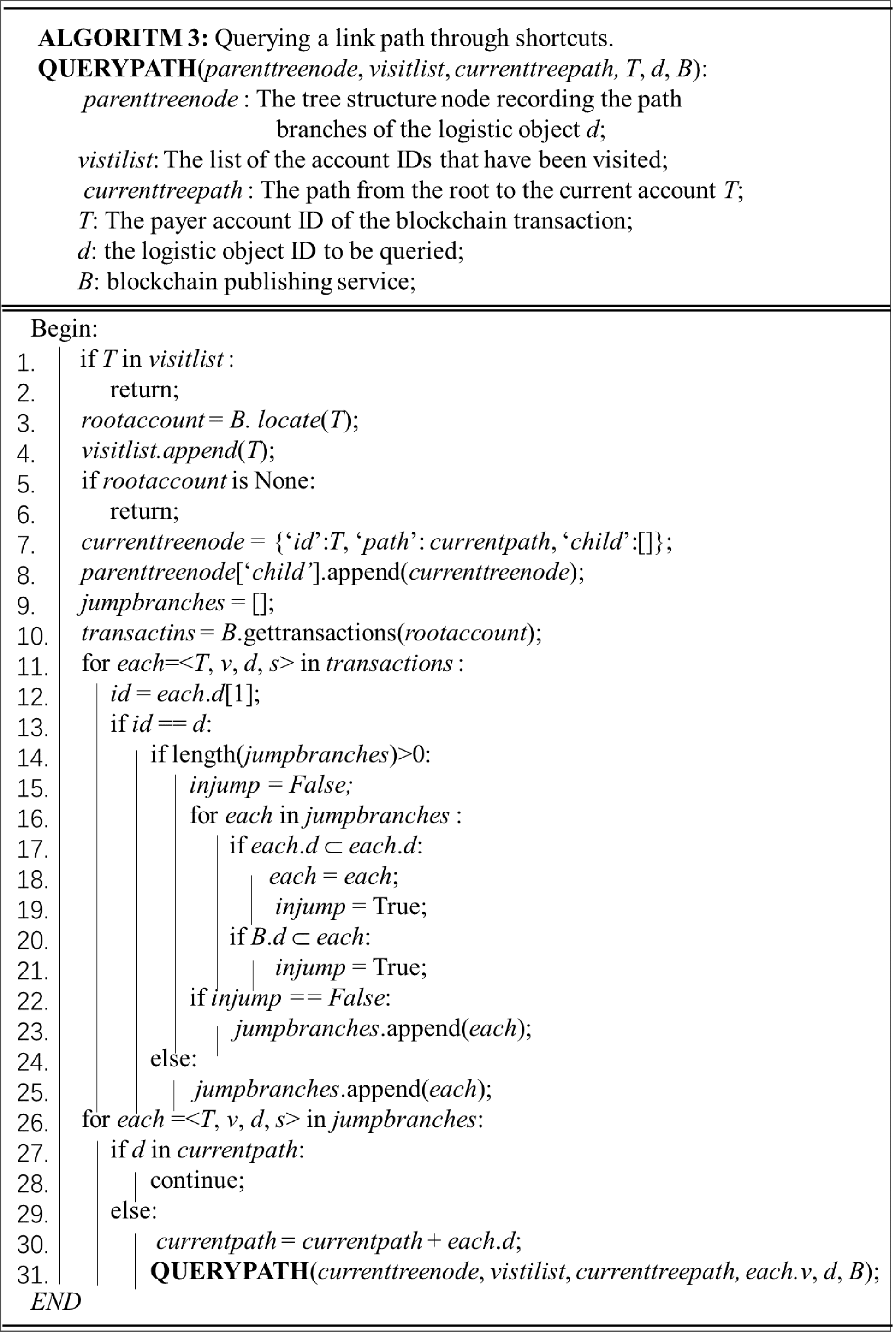}
\caption{Algorithm of querying a link path through shortcuts.}\label{fig9}
\end{figure} 

The Algorithm 3 starts from the root account T of a logistic process. In each iteration, a tree node \textit{currenttreenode} is constructed in line 8 and is set as a child node of the tree node parent tree node, which finally forms a tree consisting of paths of the logistic object $d$ has been traversed as the query result. After constructing the current tree node, all the blockchain transactions started from $T$ are obtained (line 11) and each blockchain transaction is checked to find those blockchain transaction that are used to transfer the query target $d$ (line 13 and 14). For each of such a blockchain transaction, two cases need to be handled. If the suffix ID (B.d in line 13) of the blockchain transaction $B=<T, v, d, s>$  is longer than an existing suffix IDs that have been located and store in \textit{jumpbranches} (line 18), then the blockchain transaction that stored in \textit{jumpbranches} is replaced by current $B$ (line 19). If there is no such blockchain transaction, $B$ is a new branch of transportation of d and is appended to the list \textit{jumpbranches}. In this way, the blockchain transaction that has the longest jump can be located. After checking all the blockchain transactions sent from $T$, \textit{jumpbranches} contains the longest paths that $d$ has taken. Then, for each jump in \textit{jumpbranches}, the iterative call on \textit{QUERYPATH(currenttreenode, vistilist, currenttreepath, each.v, d, B)} will recursively process each branch using the updated input argument in the last call of the procedure until reaching the leaf nodes that has no branch of transportation of $d$ (line 32). The recursive call can be processed in an asynchronized way so that the branches can be processed concurrently processed efficiently by leveraging the shortcut jumps. 
According to Theorem 2, the expected number of shortcuts added for each node is $ln(n)$, where $n$ is the total number of nodes of the link. The number of steps for finding a node is also bounded by $ln(n)$ (Theorem 3 and 4). According to Theorem 5, the number of suffix IDs for duplicated during the short cut construction along a path is $O(n^2)$. All links can be obtained within the $O(ln(n))$ steps by allowing a query from starting node to concurrently process along all the shortcut links, which can be deemed as a width-first concurrent traverse on a tree of depth $O(ln(n))$.

\begin{theorem}\label{thm2}
The expected number of shortcut links for $v_1$ is $O(ln(n))$.
\end{theorem}
\begin{proof}

At the $i$th step of movement of $d_i$, the probability of adding a shortcut link from $v_1$ is $1/i$. The expectation of the number of links from $v_1$ to the $ith$ node is also $1/i$. The expectation of the total number of shortcut links of $v1$ is the sum of expectations of the $ith$ node for $i=1..n$, which is $O(ln(n))$ according to the upper bound of the limitation of the sum of first $n$ elements in a Harmonic series \cite{RN35}.
\end{proof}

\begin{theorem}\label{thm3}
The expected number of steps for reaching a target node $vj$ from $vk$ is $O(ln(j-k))$.
\end{theorem}
\begin{proof}
Assuming $vk$ is at the $t$th step and the number of steps from $v_k$ to $v_j$ is $j-k$, which is represented as $s(t)=j-k$. Then there is one shortcut $v_q$ from $v_k$ that can jump across a segment of steps $(j-k)/e$ with high probability. That is, $v_q$ that can help reduce the distance from $1/e$ of $s(t)$ and the distance from the $v_q$ to $v_j$ is $(j-k)/e$. Recursively applying this process, there is at most $O(ln(j-k))$ steps to reach $v_j$.
\end{proof}

\begin{theorem}\label{thm4}
The expected longest jump from $v_k$ to $v_j$ is $O((j-k)/e)$.
\end{theorem}
\begin{proof}
Following Theorem 3, we can obtain the expected upper-bound of the longest jump.
\end{proof}

According to Theorem 3 and 4. The expected number of IDs stored for a node $vk$ is the $(n-k)/e$. There are $n(n-k)/e$ node IDs on n nodes for one link path and a query from any node to its following ending node in the path can be processed in $O(ln(n))$ steps involving $O(ln(n))$ access to nodes.
\begin{theorem}\label{thm5}
 The expected number of suffix IDs for a link of length $n$ is $O(n2/e))$.
\end{theorem}
\begin{proof}
 Following Theorem 2. The longest expected jump is $1/e$, which incurs $n/e$ replicas of suffix IDs.  Therefore, $n$ nodes will have at most $n^2/e$ replicas of node IDs. 
\end{proof}

\subsection{Penalty-Reward for Confirmation of Link States}\label{subsec4-6}
A penalty-reward policy is designed to encourage users to publish semantic links in a consistence way to avoid conflict in the confirmation of a link state by the participants of the logistic transportations. When publishing a semantic link u $ u-<l,d>\rightarrow v$ with state $l(n)$, the hosts of node $u$ and $v$ of the semantic link need to confirm the consistence of the published link with the real logistic transportation. If the link is correctly published, then the two hosts should reach consensus on the state of the published link. Otherwise, they should make a check on the status of the logistic transportation. If they cannot reach the consensus on the state of the link, they are unable to publish new links. There should be a penalty on the account $u$ who publishes inconsistent links that cannot be confirmed by node $v$. But there can be malicious responses on confirmation, thus we need to control the confirmation process to let responses also pay costs so that malicious responses can be suppressed.

\begin{figure}[htbp]
\centering
\includegraphics[width=0.5\textwidth]{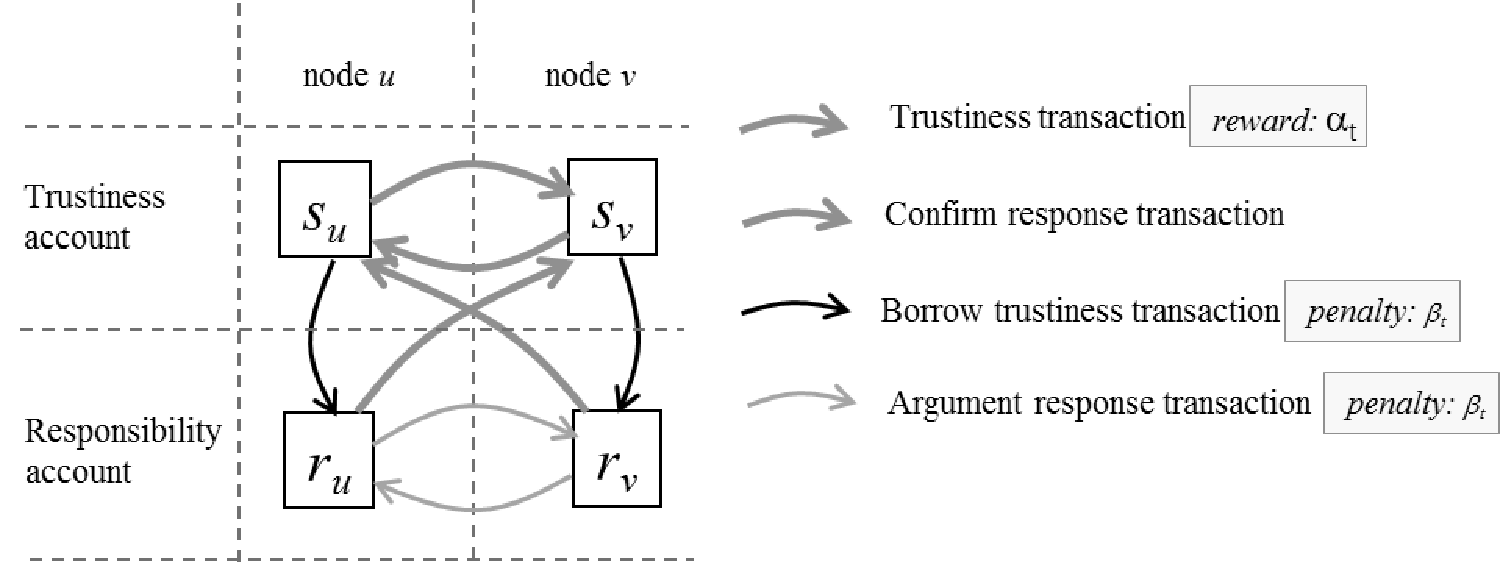}
\caption{Confirmation transactions.}\label{fig10}
\end{figure}  

To implement such a policy, each node is assigned with two scores in two accounts $su$ and $ru$, where $su$ is the trustiness account that is used for starting a confirmation of link and $ru$ is the responsibility account that is used to store the score that needs to response to a confirmation. Then, the confirmation process is conducted by transferring scores through blockchain transactions among the trustiness account and responsibility account of two nodes. It can be implemented in a smart contract. Four types of transactions can be issued among the four accounts of two nodes u and v as shown in Fig. 10. The whole confirmation procedure can be executed in following four stages:

1)	The trustiness transactions are issued between $s_u$ and $s_v$, which are used to send confirmation requests and confirmation responses for each published link $t$ (two green arrows in Fig.10). The trustiness score is added by $S_t$ when there is a new link $t$ to be confirmed. A trustiness transaction will transfer $S_t$ from $s_u$ to $s_v$. Thus, the trustiness transaction can be executed only one time for one link. When node $v$ confirms the link state, the trustiness score $S_t$ is sent back to $s_u$ in another trustiness transaction and a reward score $\alpha_t$ is added to both trustiness accounts of $s_u$ and $s_v$.
2)	If node $v$ has a confliction on the confirmation, it needs to start an argument to response node $u$. In this case, it will not use the trustiness transaction. Instead, node v uses the account $r_v$ to make an argument response transaction to $r_u$ with score $S_t$ (gray arrow in Fig. 10). However, the response account $r_v$ does not have initial scores, which needs to borrow scores from its own trustiness account $sv$ in a borrow trustiness transaction (dark arrows in Fig. 10). Then, the score is sent to the responsibility account in an argument response transaction. When node $u$ receives the transaction scores in its $r_u$, it knows that there is a confliction on confirmation. Node $u$ needs to check the transportation with node $v$ so that they can reach the consensus on the state of link.
3)	If they reach a consensus, node $u$ uses the confirm response transaction to send score $S_t$ back to the trustiness account $s_v$ of node $v$. Then, node $v$ sends $S_t$ back to the node $u$ through a trustiness transaction. In this case, the confirmation completes.
4)	If they need further arguments, response transactions are used to send scores from their own responsibility account to counterpart responsibility account, until one sends the score $S_t$ back to the trustiness account, which indicates that they reach consensus.
Since the responsibility account initially has a zero score, it needs to borrow scores from the trustiness account to make an argument. Moreover, there is an extra penalty $\alpha_t$ charged on their trustiness account for each borrow trustiness transaction. Each argument transaction is also charged by $\alpha_t$. That is, if more arguments are made, their trustiness score will be exhausted, and they will have not enough trustiness score to start a new confirmation process for a new link and start an extra argument response. So, malicious responses are costly for both sides of a link. On the contrary, if they can reach consensus quickly on a confirmation just using the trustiness transactions, rewards can be accumulated in their trustiness accounts.
Finally, a reliability score $R_u=s_u-r_u$ can be obtained for a node u as the indicator of the performance of $u$. A higher Ru means that there are more consistent links published and less arguments for node $u$. The nodes with reliability score lower than a threshold R will be halt.
As the topology shown in Fig. 10 is strong connected, the scores will finally approach an ending distribution for four accounts. In an extreme case, the two nodes are halt after their reliability score is lower than $R$.

\begin{theorem}\label{thm6}
The scores of $s_u$, $s_v$, $r_u$, and $r_v$ will approach to a halting condition for one complete confirmation process with two trustiness transactions.
\end{theorem}
\begin{proof}
The graph of Fig. 10 is strongly connected. There are two possible ending states: two trustiness transactions are finally executed; or their reliability scores $T$ are exhausted.
\end{proof}

One possible situation is that one participant does not do any further action. In this case, its responsibility account and trustiness account are not balanced. The further actions can be halt according to the unbalance of the two accounts of a host.
All the transactions are published on the blockchain and the two accounts are opened and signed by their owners so that anyone can obtain the reliability score of a host by visiting the opened transactions and corresponding trustiness and responsibility accounts. Although short cut links can increase the amount of transactions published on blockchain, it is still controlled and relatively much smaller compared with real transactions published on blockchain for bitcoin transactions. Rather than using smart contract, the logistic transportation data is encapsulated in transactions in a total decentralized way, which is much more secure than using smart contract based solution that requires a central implementation of data management.

The SLN is used to represent logistic transportations and are published on the blockchain platform in a decentralized way to support decentralized traceability of the states of logistic transportations. It has a good expressiveness in representing the locations of the logistic objects and states of logistic transportations in a logistic process. By establishing the mapping between the states of links and the blockchain transactions, decentralized publishing and tracing of links and their states are implemented on the blockchain. The method for constructing shortcuts on link paths achieves logarithmic number of accesses to nodes for obtaining links on a logistic transportation path. The effectiveness of the proposed method is further verified by simulation experiments. A reward-penalty policy is designed for two participants to reach consensus on the confirmation of the state of links.

The SLN provides a general model for developing various domain applications by defining the schema of SLN and the schema of state transitions.  Users can query the changing states on nodes and links and efficiently trace the changes through user interface without the need to concern the storage of data at the physical level and logical level.  The function of publishing and querying links on blockchain can be encapsulated as basic graph data management operators to support implementation of decentralized graph data applications on blockchain, enabling developers to focus on high-level graph data management through the developers’ interface without knowing the details of implementing the complex data model manipulation on blockchain platform.  The SLN has extensibility for implementing wide applications that need node and path querying on a semantically rich network based on reasoning on semantic links.

\begin{figure}[htbp]
\centering
\includegraphics[width=0.5\textwidth]{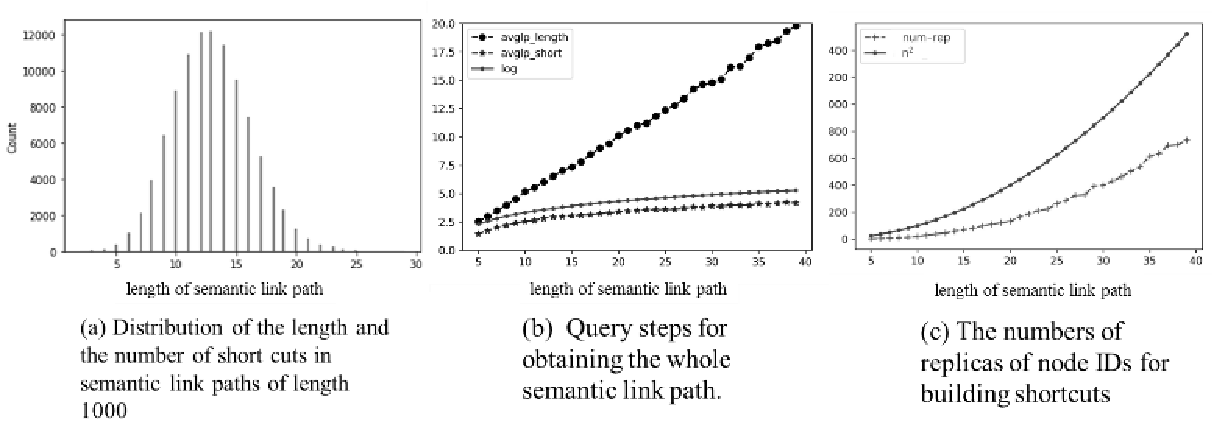}
\caption{Distribution of the length and the number of short cuts on semantic link paths of length 1000.}\label{fig11}
\end{figure}  

\section{Experimental Evaluation of Path State Query}\label{sec5}
Simulation experiments are carried out to evaluate the path query on the blockchain. We randomly generated a set of logistic processes, each consisting of a set of logistic objects that are to be transferred among a group of parties. Then, we evaluate the average costs of obtaining a link within a path of links by using the shortcuts method.
Fig. 11a shows the distribution of the number of shortcuts of nodes on a path of length 1000. It can be seen that both the average number and the max number of shortcuts for one node is bounded within a range from 1 to 25. That is, although the number $n$ of total nodes is unknown, the number of shortcuts for one node is still in control. Fig. 11b shows the steps of obtaining all nodes of one path. Each test case has 1000 paths with the average length from 5 to 40.  The $x$-axis is the average length of the paths. When there is no shortcut, the average number of steps is in linear growth with the length of path (shown by curve $avglp\_length$ in Fig.11b). When shortcuts are added to each path, the steps to obtain the whole path is within logarithms scale with respects to the length of path (shown by $avglp\_short$ curve in Fig. 8b), which is below the theoretic bound curve (log in Fig. 11b). Fig. 11c shows the number of replicas of node IDs produced after adding shortcuts to link paths. The number of replicas ($num\-rep$ in Fig. 11c) is bounded by $O(n^2)$ (shown as $n2$ in Fig. 11c). In summary, the simulation results confirm the analysis of the querying efficiency of paths when adding shortcuts to the path.
 
\section{Conclusion}\label{sec6}
The SLN is used to represent logistic transportations and are published on the blockchain platform in a decentralized way to support decentralized traceability of the states of logistic transportations. It has a good expressiveness in representing the locations of the logistic objects and states of logistic transportations in a logistic process. By establishing the mapping between the states of links and the blockchain transactions, decentralized publishing and tracing of links and their states are implemented on the blockchain. The method for constructing shortcuts on link paths achieves logarithmic number of accesses to nodes for obtaining links on a logistic transportation path. The effectiveness of the proposed method is further verified by simulation experiments. A reward-penalty policy is designed for two participants to reach consensus on the confirmation of the state of links.

The SLN provides a general model for developing various domain applications by defining the schema of SLN and the schema of state transitions.  Users can query the changing states on nodes and links and efficiently trace the changes through user interface without the need to concern the storage of data at the physical level and logical level.  The function of publishing and querying links on blockchain can be encapsulated as basic graph data management operators to support implementation of decentralized graph data applications on blockchain, enabling developers to focus on high-level graph data management through the developers’ interface without knowing the details of implementing the complex data model manipulation on blockchain platform.  The SLN has extensibility for implementing wide applications that need node and path querying on a semantically rich network based on reasoning on semantic links.

\backmatter
\bmhead{Acknowledgements}

This work was supported by National Science Foundation of China (project no. 61876048).  Hai Zhuge is the correspondence author.
\bigskip 

\bibliography{sn-bibliography}


\begin{thebibliography}{38}
\ifx \bisbn   \undefined \def \bisbn  #1{ISBN #1}\fi
\ifx \binits  \undefined \def \binits#1{#1}\fi
\ifx \bauthor  \undefined \def \bauthor#1{#1}\fi
\ifx \batitle  \undefined \def \batitle#1{#1}\fi
\ifx \bjtitle  \undefined \def \bjtitle#1{#1}\fi
\ifx \bvolume  \undefined \def \bvolume#1{\textbf{#1}}\fi
\ifx \byear  \undefined \def \byear#1{#1}\fi
\ifx \bissue  \undefined \def \bissue#1{#1}\fi
\ifx \bfpage  \undefined \def \bfpage#1{#1}\fi
\ifx \blpage  \undefined \def \blpage #1{#1}\fi
\ifx \burl  \undefined \def \burl#1{\textsf{#1}}\fi
\ifx \doiurl  \undefined \def \doiurl#1{\url{https://doi.org/#1}}\fi
\ifx \betal  \undefined \def \betal{\textit{et al.}}\fi
\ifx \binstitute  \undefined \def \binstitute#1{#1}\fi
\ifx \binstitutionaled  \undefined \def \binstitutionaled#1{#1}\fi
\ifx \bctitle  \undefined \def \bctitle#1{#1}\fi
\ifx \beditor  \undefined \def \beditor#1{#1}\fi
\ifx \bpublisher  \undefined \def \bpublisher#1{#1}\fi
\ifx \bbtitle  \undefined \def \bbtitle#1{#1}\fi
\ifx \bedition  \undefined \def \bedition#1{#1}\fi
\ifx \bseriesno  \undefined \def \bseriesno#1{#1}\fi
\ifx \blocation  \undefined \def \blocation#1{#1}\fi
\ifx \bsertitle  \undefined \def \bsertitle#1{#1}\fi
\ifx \bsnm \undefined \def \bsnm#1{#1}\fi
\ifx \bsuffix \undefined \def \bsuffix#1{#1}\fi
\ifx \bparticle \undefined \def \bparticle#1{#1}\fi
\ifx \barticle \undefined \def \barticle#1{#1}\fi
\bibcommenthead
\ifx \bconfdate \undefined \def \bconfdate #1{#1}\fi
\ifx \botherref \undefined \def \botherref #1{#1}\fi
\ifx \url \undefined \def \url#1{\textsf{#1}}\fi
\ifx \bchapter \undefined \def \bchapter#1{#1}\fi
\ifx \bbook \undefined \def \bbook#1{#1}\fi
\ifx \bcomment \undefined \def \bcomment#1{#1}\fi
\ifx \oauthor \undefined \def \oauthor#1{#1}\fi
\ifx \citeauthoryear \undefined \def \citeauthoryear#1{#1}\fi
\ifx \endbibitem  \undefined \def \endbibitem {}\fi
\ifx \bconflocation  \undefined \def \bconflocation#1{#1}\fi
\ifx \arxivurl  \undefined \def \arxivurl#1{\textsf{#1}}\fi
\csname PreBibitemsHook\endcsname

\bibitem[\protect\citeauthoryear{Nakamoto}{2008}]{RN1}
\begin{botherref}
\oauthor{\bsnm{Nakamoto}, \binits{S.}}:
Bitcoin: a Peer-to-Peer Electronic Cash System
(2008).
\url{http://bitcoin.org/bitcoin.pdf}
\end{botherref}
\endbibitem

\bibitem[\protect\citeauthoryear{McConaghy et~al.}{2016}]{RN2}
\begin{botherref}
\oauthor{\bsnm{McConaghy}, \binits{T.}},
\oauthor{\bsnm{Marques}, \binits{R.}},
\oauthor{\bsnm{Müller}, \binits{A.}},
\oauthor{\bsnm{Jonghe}, \binits{D.D.}},
\oauthor{\bsnm{McConaghy}, \binits{T.}},
\oauthor{\bsnm{McMullen}, \binits{G.}},
\oauthor{\bsnm{Henderson}, \binits{R.}},
\oauthor{\bsnm{Bellemare}, \binits{S.}},
\oauthor{\bsnm{Granzotto}, \binits{A.}}:
BigchainDB: a scalable blockchain database
(2016).
\url{https://git.berlin/bigchaindb/site/raw/commit/b2d98401b65175f0fe0c169932ddca0b98a456a6/_src/whitepaper/bigchaindb-whitepaper.pdf.}
\end{botherref}
\endbibitem

\bibitem[\protect\citeauthoryear{Choi and Siqin}{2022}]{RN3}
\begin{barticle}
\bauthor{\bsnm{Choi}, \binits{T.-M.}},
\bauthor{\bsnm{Siqin}, \binits{T.}}:
\batitle{Blockchain in logistics and production from blockchain 1.0 to
  blockchain 5.0: An intra-inter-organizational framework}.
\bjtitle{Transportation Research Part E: Logistics and Transportation Review}
\bvolume{160},
\bfpage{102653}
(\byear{2022})
\doiurl{10.1016/j.tre.2022.102653}
\end{barticle}
\endbibitem

\bibitem[\protect\citeauthoryear{Politou et~al.}{2021}]{RN4}
\begin{barticle}
\bauthor{\bsnm{Politou}, \binits{E.}},
\bauthor{\bsnm{Casino}, \binits{F.}},
\bauthor{\bsnm{Alepis}, \binits{E.}},
\bauthor{\bsnm{Patsakis}, \binits{C.}}:
\batitle{Blockchain mutability: Challenges and proposed solutions}.
\bjtitle{IEEE Transactions on Emerging Topics in Computing}
\bvolume{9}(\bissue{4}),
\bfpage{1972}--\blpage{1986}
(\byear{2021})
\doiurl{10.1109/TETC.2019.2949510}
\end{barticle}
\endbibitem

\bibitem[\protect\citeauthoryear{Guo et~al.}{2023}]{RN5}
\begin{barticle}
\bauthor{\bsnm{Guo}, \binits{L.}},
\bauthor{\bsnm{Wang}, \binits{Q.}},
\bauthor{\bsnm{Yau}, \binits{W.C.}}:
\batitle{Online/offline rewritable blockchain with auditable outsourced
  computation}.
\bjtitle{IEEE Transactions on Cloud Computing}
\bvolume{11}(\bissue{1}),
\bfpage{499}--\blpage{514}
(\byear{2023})
\doiurl{10.1109/TCC.2021.3102031}
\end{barticle}
\endbibitem

\bibitem[\protect\citeauthoryear{López-Pintado et~al.}{}]{RN6}
\begin{botherref}
\oauthor{\bsnm{López-Pintado}, \binits{O.}},
\oauthor{\bsnm{Dumas}, \binits{M.}},
\oauthor{\bsnm{García-Bañuelos}, \binits{L.}},
\oauthor{\bsnm{Weber}, \binits{I.}}:
Interpreted execution of business process models on blockchain.
In: 2019 IEEE 23rd International Enterprise Distributed Object Computing
  Conference (EDOC),
pp. 206--215.
\doiurl{10.1109/EDOC.2019.00033}
\end{botherref}
\endbibitem

\bibitem[\protect\citeauthoryear{Zhuge}{2009}]{RN7}
\begin{barticle}
\bauthor{\bsnm{Zhuge}, \binits{H.}}:
\batitle{Communities and emerging semantics in semantic link network: Discovery
  and learning}.
\bjtitle{IEEE Transactions on Knowledge and Data Engineering}
\bvolume{21}(\bissue{6}),
\bfpage{785}--\blpage{799}
(\byear{2009})
\end{barticle}
\endbibitem

\bibitem[\protect\citeauthoryear{Zhao et~al.}{2019}]{RN8}
\begin{barticle}
\bauthor{\bsnm{Zhao}, \binits{K.}},
\bauthor{\bsnm{Scheibe}, \binits{K.}},
\bauthor{\bsnm{Blackhurst}, \binits{J.}},
\bauthor{\bsnm{Kumar}, \binits{A.}}:
\batitle{Supply chain network robustness against disruptions: Topological
  analysis, measurement, and optimization}.
\bjtitle{IEEE Transactions on Engineering Management}
\bvolume{66}(\bissue{1}),
\bfpage{127}--\blpage{139}
(\byear{2019})
\doiurl{10.1109/TEM.2018.2808331}
\end{barticle}
\endbibitem

\bibitem[\protect\citeauthoryear{Appel et~al.}{2014}]{RN9}
\begin{barticle}
\bauthor{\bsnm{Appel}, \binits{S.}},
\bauthor{\bsnm{Kleber}, \binits{P.}},
\bauthor{\bsnm{Frischbier}, \binits{S.}},
\bauthor{\bsnm{Freudenreich}, \binits{T.}},
\bauthor{\bsnm{Buchmann}, \binits{A.}}:
\batitle{Modeling and execution of event stream processing in business
  processes}.
\bjtitle{Information Systems}
\bvolume{46},
\bfpage{140}--\blpage{156}
(\byear{2014})
\doiurl{10.1016/j.is.2014.04.002}
\end{barticle}
\endbibitem

\bibitem[\protect\citeauthoryear{Combi et~al.}{2023}]{RN10}
\begin{barticle}
\bauthor{\bsnm{Combi}, \binits{C.}},
\bauthor{\bsnm{Oliboni}, \binits{B.}},
\bauthor{\bsnm{Zerbato}, \binits{F.}}:
\batitle{Integrated exploration of data-intensive business processes}.
\bjtitle{IEEE Transactions on Services Computing}
\bvolume{16}(\bissue{1}),
\bfpage{383}--\blpage{397}
(\byear{2023})
\doiurl{10.1109/TSC.2021.3134485}
\end{barticle}
\endbibitem

\bibitem[\protect\citeauthoryear{Biswas et~al.}{2022}]{RN11}
\begin{barticle}
\bauthor{\bsnm{Biswas}, \binits{D.}},
\bauthor{\bsnm{Jalali}, \binits{H.}},
\bauthor{\bsnm{Ansaripoor}, \binits{A.H.}},
\bauthor{\bsnm{De~Giovanni}, \binits{P.}}:
\batitle{Traceability vs. sustainability in supply chains: The implications of
  blockchain}.
\bjtitle{European Journal of Operational Research}
\bvolume{305}(\bissue{1}),
\bfpage{128}--\blpage{147}
(\byear{2022})
\doiurl{10.1016/j.ejor.2022.05.034}
\end{barticle}
\endbibitem

\bibitem[\protect\citeauthoryear{Menon and Jain}{2024}]{RN12}
\begin{barticle}
\bauthor{\bsnm{Menon}, \binits{S.}},
\bauthor{\bsnm{Jain}, \binits{K.}}:
\batitle{Blockchain technology for transparency in agri-food supply chain: Use
  cases, limitations, and future directions}.
\bjtitle{IEEE Transactions on Engineering Management}
\bvolume{71},
\bfpage{106}--\blpage{120}
(\byear{2024})
\doiurl{10.1109/TEM.2021.3110903}
\end{barticle}
\endbibitem

\bibitem[\protect\citeauthoryear{Munasinghe and Halgamuge}{2023}]{RN13}
\begin{barticle}
\bauthor{\bsnm{Munasinghe}, \binits{U.J.}},
\bauthor{\bsnm{Halgamuge}, \binits{M.N.}}:
\batitle{Supply chain traceability and counterfeit detection of covid-19
  vaccines using novel blockchain-based vacledger system}.
\bjtitle{Expert Systems with Applications}
\bvolume{228},
\bfpage{120293}
(\byear{2023})
\doiurl{10.1016/j.eswa.2023.120293}
\end{barticle}
\endbibitem

\bibitem[\protect\citeauthoryear{Dinh et~al.}{2018}]{RN14}
\begin{barticle}
\bauthor{\bsnm{Dinh}, \binits{T.T.A.}},
\bauthor{\bsnm{Liu}, \binits{R.}},
\bauthor{\bsnm{Zhang}, \binits{M.}},
\bauthor{\bsnm{Chen}, \binits{G.}},
\bauthor{\bsnm{Ooi}, \binits{B.C.}},
\bauthor{\bsnm{Wang}, \binits{J.}}:
\batitle{Untangling blockchain: A data processing view of blockchain systems}.
\bjtitle{IEEE Transactions on Knowledge and Data Engineering}
\bvolume{30}(\bissue{7}),
\bfpage{1366}--\blpage{1385}
(\byear{2018})
\doiurl{10.1109/TKDE.2017.2781227}
\end{barticle}
\endbibitem

\bibitem[\protect\citeauthoryear{Cui et~al.}{2022}]{RN15}
\begin{barticle}
\bauthor{\bsnm{Cui}, \binits{J.}},
\bauthor{\bsnm{Ouyang}, \binits{F.}},
\bauthor{\bsnm{Ying}, \binits{Z.}},
\bauthor{\bsnm{Wei}, \binits{L.}},
\bauthor{\bsnm{Zhong}, \binits{H.}}:
\batitle{Secure and efficient data sharing among vehicles based on consortium
  blockchain}.
\bjtitle{IEEE Transactions on Intelligent Transportation Systems}
\bvolume{23}(\bissue{7}),
\bfpage{8857}--\blpage{8867}
(\byear{2022})
\doiurl{10.1109/TITS.2021.3086976}
\end{barticle}
\endbibitem

\bibitem[\protect\citeauthoryear{Li et~al.}{2020}]{RN16}
\begin{barticle}
\bauthor{\bsnm{Li}, \binits{X.}},
\bauthor{\bsnm{Jiang}, \binits{P.}},
\bauthor{\bsnm{Chen}, \binits{T.}},
\bauthor{\bsnm{Luo}, \binits{X.}},
\bauthor{\bsnm{Wen}, \binits{Q.}}:
\batitle{A survey on the security of blockchain systems}.
\bjtitle{Future Generation Computer Systems}
\bvolume{107},
\bfpage{841}--\blpage{853}
(\byear{2020})
\doiurl{10.1016/j.future.2017.08.020}
\end{barticle}
\endbibitem

\bibitem[\protect\citeauthoryear{Salman et~al.}{2019}]{RN17}
\begin{barticle}
\bauthor{\bsnm{Salman}, \binits{T.}},
\bauthor{\bsnm{Zolanvari}, \binits{M.}},
\bauthor{\bsnm{Erbad}, \binits{A.}},
\bauthor{\bsnm{Jain}, \binits{R.}},
\bauthor{\bsnm{Samaka}, \binits{M.}}:
\batitle{Security services using blockchains: A state of the art survey}.
\bjtitle{IEEE Communications Surveys and Tutorials}
\bvolume{21}(\bissue{1}),
\bfpage{858}--\blpage{880}
(\byear{2019})
\doiurl{10.1109/COMST.2018.2863956}
\end{barticle}
\endbibitem

\bibitem[\protect\citeauthoryear{Amiri et~al.}{}]{RN18}
\begin{botherref}
\oauthor{\bsnm{Amiri}, \binits{M.J.}},
\oauthor{\bsnm{Agrawal}, \binits{D.}},
\oauthor{\bsnm{Abbadi}, \binits{A.E.}}:
Permissioned blockchains: Properties, techniques and applications.
In: Proceedings of the 2021 International Conference on Management of Data,
pp. 2813--2820.
Association for Computing Machinery.
\doiurl{10.1145/3448016.3457539} .
\url{https://doi.org/10.1145/3448016.3457539}
\end{botherref}
\endbibitem

\bibitem[\protect\citeauthoryear{Amiri et~al.}{}]{RN19}
\begin{botherref}
\oauthor{\bsnm{Amiri}, \binits{M.J.}},
\oauthor{\bsnm{Agrawal}, \binits{D.}},
\oauthor{\bsnm{Abbadi}, \binits{A.E.}}:
Sharper: Sharding permissioned blockchains over network clusters.
In: Proceedings of the 2021 International Conference on Management of Data,
pp. 76--88.
Association for Computing Machinery.
\doiurl{10.1145/3448016.3452807} .
\url{https://doi.org/10.1145/3448016.3452807}
\end{botherref}
\endbibitem

\bibitem[\protect\citeauthoryear{Wang et~al.}{2020}]{RN21}
\begin{barticle}
\bauthor{\bsnm{Wang}, \binits{Z.}},
\bauthor{\bsnm{Wang}, \binits{T.}},
\bauthor{\bsnm{Hu}, \binits{H.}},
\bauthor{\bsnm{Gong}, \binits{J.}},
\bauthor{\bsnm{Ren}, \binits{X.}},
\bauthor{\bsnm{Xiao}, \binits{Q.}}:
\batitle{Blockchain-based framework for improving supply chain traceability and
  information sharing in precast construction}.
\bjtitle{Automation in Construction}
\bvolume{111},
\bfpage{103063}
(\byear{2020})
\doiurl{10.1016/j.autcon.2019.103063}
\end{barticle}
\endbibitem

\bibitem[\protect\citeauthoryear{Liu and Li}{2020}]{RN22}
\begin{barticle}
\bauthor{\bsnm{Liu}, \binits{Z.}},
\bauthor{\bsnm{Li}, \binits{Z.}}:
\batitle{A blockchain-based framework of cross-border e-commerce supply chain}.
\bjtitle{International Journal of Information Management}
\bvolume{52},
\bfpage{102059}
(\byear{2020})
\doiurl{10.1016/j.ijinfomgt.2019.102059}
\end{barticle}
\endbibitem

\bibitem[\protect\citeauthoryear{Hader et~al.}{2022}]{RN23}
\begin{barticle}
\bauthor{\bsnm{Hader}, \binits{M.}},
\bauthor{\bsnm{Tchoffa}, \binits{D.}},
\bauthor{\bsnm{Mhamedi}, \binits{A.E.}},
\bauthor{\bsnm{Ghodous}, \binits{P.}},
\bauthor{\bsnm{Dolgui}, \binits{A.}},
\bauthor{\bsnm{Abouabdellah}, \binits{A.}}:
\batitle{Applying integrated blockchain and big data technologies to improve
  supply chain traceability and information sharing in the textile sector}.
\bjtitle{Journal of Industrial Information Integration}
\bvolume{28},
\bfpage{100345}
(\byear{2022})
\doiurl{10.1016/j.jii.2022.100345}
\end{barticle}
\endbibitem

\bibitem[\protect\citeauthoryear{Zarrin et~al.}{2021}]{RN24}
\begin{barticle}
\bauthor{\bsnm{Zarrin}, \binits{J.}},
\bauthor{\bsnm{Wen~Phang}, \binits{H.}},
\bauthor{\bsnm{Babu~Saheer}, \binits{L.}},
\bauthor{\bsnm{Zarrin}, \binits{B.}}:
\batitle{Blockchain for decentralization of internet: prospects, trends, and
  challenges}.
\bjtitle{Cluster Computing}
\bvolume{24}(\bissue{4}),
\bfpage{2841}--\blpage{2866}
(\byear{2021})
\doiurl{10.1007/s10586-021-03301-8}
\end{barticle}
\endbibitem

\bibitem[\protect\citeauthoryear{Ng}{}]{RN25}
\begin{botherref}
\oauthor{\bsnm{Ng}, \binits{W.}}:
Developing rfid database models for analysing moving tags in supply chain
  management.
In: In Proceedings of the 30th International Conference on Conceptual Modeling
  (ER'11).,
pp. 204--218
\end{botherref}
\endbibitem

\bibitem[\protect\citeauthoryear{Lu and Xu}{2017}]{RN26}
\begin{barticle}
\bauthor{\bsnm{Lu}, \binits{Q.}},
\bauthor{\bsnm{Xu}, \binits{X.}}:
\batitle{Adaptable blockchain-based systems: A case study for product
  traceability}.
\bjtitle{IEEE Software}
\bvolume{34}(\bissue{6}),
\bfpage{21}--\blpage{27}
(\byear{2017})
\doiurl{10.1109/MS.2017.4121227}
\end{barticle}
\endbibitem

\bibitem[\protect\citeauthoryear{López-Pintado et~al.}{}]{RN27}
\begin{botherref}
\oauthor{\bsnm{López-Pintado}, \binits{O.}},
\oauthor{\bsnm{García-Bañuelos}, \binits{L.}},
\oauthor{\bsnm{Dumas}, \binits{M.}},
\oauthor{\bsnm{Weber}, \binits{I.}}:
Caterpillar: A blockchain-based business process management system.
In: International Conference on Business Process Management
\end{botherref}
\endbibitem

\bibitem[\protect\citeauthoryear{Zhu et~al.}{2023}]{RN28}
\begin{barticle}
\bauthor{\bsnm{Zhu}, \binits{P.}},
\bauthor{\bsnm{Hu}, \binits{J.}},
\bauthor{\bsnm{Li}, \binits{X.}},
\bauthor{\bsnm{Zhu}, \binits{Q.}}:
\batitle{Using blockchain technology to enhance the traceability of original
  achievements}.
\bjtitle{IEEE Transactions on Engineering Management}
\bvolume{70}(\bissue{5}),
\bfpage{1693}--\blpage{1707}
(\byear{2023})
\doiurl{10.1109/TEM.2021.3066090}
\end{barticle}
\endbibitem

\bibitem[\protect\citeauthoryear{Zhuge}{2012}]{RN29}
\begin{bbook}
\bauthor{\bsnm{Zhuge}, \binits{H.}}:
\bbtitle{“Semantic Link Network” in The Knowledge Grid: Toward
  Cyber-Physical Society}.
\bpublisher{World Scientific Publishing Co}, \blocation{???}
(\byear{2012})
\end{bbook}
\endbibitem

\bibitem[\protect\citeauthoryear{Zhuge}{2020}]{RN30}
\begin{bbook}
\bauthor{\bsnm{Zhuge}, \binits{H.}}:
\bbtitle{Cyber-Physical-Social Semantic Link Network},
pp. \bfpage{55}--\blpage{141}.
\bpublisher{Springer},
\blocation{Singapore}
(\byear{2020}).
\doiurl{10.1007/978-981-13-7311-4_3} .
\burl{https://doi.org/10.1007/978-981-13-7311-4-3}
\end{bbook}
\endbibitem

\bibitem[\protect\citeauthoryear{Zhuge}{2016}]{RN31}
\begin{bbook}
\bauthor{\bsnm{Zhuge}, \binits{H.}}:
\bbtitle{Multi-dimensional Summarization in Cyber-physical Society}.
\bpublisher{Morgan Kaufmann},
\blocation{Netherlands, UK and USA}
(\byear{2016})
\end{bbook}
\endbibitem

\bibitem[\protect\citeauthoryear{Ikkai et~al.}{2000}]{RN32}
\begin{barticle}
\bauthor{\bsnm{Ikkai}, \binits{Y.}},
\bauthor{\bsnm{Maruta}, \binits{T.}},
\bauthor{\bsnm{Komoda}, \binits{N.}},
\bauthor{\bsnm{Goossenaerts}, \binits{J.}}:
\batitle{A graph-base supply chain simulation language and tool with
  multi-functional modeling}.
\bjtitle{IFAC Proceedings Volumes}
\bvolume{33}(\bissue{17}),
\bfpage{935}--\blpage{940}
(\byear{2000})
\doiurl{10.1016/S1474-6670(17)39529-0}
\end{barticle}
\endbibitem

\bibitem[\protect\citeauthoryear{Rosenfeld et~al.}{2009}]{RN33}
\begin{barticle}
\bauthor{\bsnm{Rosenfeld}, \binits{A.}},
\bauthor{\bsnm{Goldman}, \binits{C.V.}},
\bauthor{\bsnm{Kaminka}, \binits{G.A.}},
\bauthor{\bsnm{Kraus}, \binits{S.}}:
\batitle{Phirst: A distributed architecture for p2p information retrieval}.
\bjtitle{Information Systems}
\bvolume{34}(\bissue{2}),
\bfpage{290}--\blpage{303}
(\byear{2009})
\doiurl{10.1016/j.is.2008.08.002}
\end{barticle}
\endbibitem

\bibitem[\protect\citeauthoryear{Wu et~al.}{2023}]{RN34}
\begin{barticle}
\bauthor{\bsnm{Wu}, \binits{H.}},
\bauthor{\bsnm{Jiang}, \binits{S.}},
\bauthor{\bsnm{Cao}, \binits{J.}}:
\batitle{High-efficiency blockchain-based supply chain traceability}.
\bjtitle{IEEE Transactions on Intelligent Transportation Systems}
\bvolume{24}(\bissue{4}),
\bfpage{3748}--\blpage{3758}
(\byear{2023})
\doiurl{10.1109/TITS.2022.3205445}
\end{barticle}
\endbibitem

\bibitem[\protect\citeauthoryear{Zhuge et~al.}{2008}]{RN35}
\begin{barticle}
\bauthor{\bsnm{Zhuge}, \binits{H.}},
\bauthor{\bsnm{Chen}, \binits{X.}},
\bauthor{\bsnm{Sun}, \binits{X.}},
\bauthor{\bsnm{Yao}, \binits{E.}}:
\batitle{Hring: A structured p2p overlay based on harmonic series}.
\bjtitle{IEEE Trans. Parallel Distrib. Syst.}
\bvolume{19}(\bissue{2}),
\bfpage{145}--\blpage{158}
(\byear{2008})
\doiurl{10.1109/TPDS.2007.70725}
\end{barticle}
\endbibitem

\bibitem[\protect\citeauthoryear{Owe and Fazeldehkordi}{2022}]{RN36}
\begin{barticle}
\bauthor{\bsnm{Owe}, \binits{O.}},
\bauthor{\bsnm{Fazeldehkordi}, \binits{E.}}:
\batitle{A lightweight approach to smart contracts supporting safety, security,
  and privacy}.
\bjtitle{Journal of Logical and Algebraic Methods in Programming}
\bvolume{127},
\bfpage{100772}
(\byear{2022})
\doiurl{10.1016/j.jlamp.2022.100772}
\end{barticle}
\endbibitem

\bibitem[\protect\citeauthoryear{Song et~al.}{2021}]{RN37}
\begin{barticle}
\bauthor{\bsnm{Song}, \binits{Q.}},
\bauthor{\bsnm{Chen}, \binits{Y.}},
\bauthor{\bsnm{Zhong}, \binits{Y.}},
\bauthor{\bsnm{Lan}, \binits{K.}},
\bauthor{\bsnm{Fong}, \binits{S.}},
\bauthor{\bsnm{Tang}, \binits{R.}}:
\batitle{A supply-chain system framework based on internet of things using
  blockchain technology}.
\bjtitle{ACM Trans. Internet Technol.}
\bvolume{21}(\bissue{1}),
\bfpage{13}
(\byear{2021})
\doiurl{10.1145/3409798}
\end{barticle}
\endbibitem

\bibitem[\protect\citeauthoryear{Li et~al.}{2021}]{RN38}
\begin{barticle}
\bauthor{\bsnm{Li}, \binits{C.}},
\bauthor{\bsnm{Zhang}, \binits{J.}},
\bauthor{\bsnm{Yang}, \binits{X.}},
\bauthor{\bsnm{Youlong}, \binits{L.}}:
\batitle{Lightweight blockchain consensus mechanism and storage optimization
  for resource-constrained iot devices}.
\bjtitle{Information Processing and Management}
\bvolume{58}(\bissue{4}),
\bfpage{102602}
(\byear{2021})
\doiurl{10.1016/j.ipm.2021.102602}
\end{barticle}
\endbibitem

\bibitem[\protect\citeauthoryear{Bevilacqua et~al.}{2009}]{RN39}
\begin{barticle}
\bauthor{\bsnm{Bevilacqua}, \binits{M.}},
\bauthor{\bsnm{Ciarapica}, \binits{F.E.}},
\bauthor{\bsnm{Giacchetta}, \binits{G.}}:
\batitle{Business process reengineering of a supply chain and a traceability
  system: A case study}.
\bjtitle{Journal of Food Engineering}
\bvolume{93}(\bissue{1}),
\bfpage{13}--\blpage{22}
(\byear{2009})
\doiurl{10.1016/j.jfoodeng.2008.12.020}
\end{barticle}
\endbibitem

\end{thebibliography}

\end{document}